\documentclass[onecolumn]{IEEEtran}
\usepackage{amsthm, amsmath, amssymb, amsfonts, url, booktabs, tikz, setspace, fancyhdr, bm}
\usepackage{hyperref}
\usepackage{enumerate}
\usepackage{cite}
\usepackage[shortlabels]{enumitem}
\usepackage[babel]{microtype}
\usepackage[english]{babel}
\usepackage[capitalise]{cleveref}
\usepackage{comment}
\usepackage{bbm}
\usepackage{csquotes}
\usepackage{amsthm}
\usepackage{mathabx}
\usepackage{tikz}
\usepackage{graphicx}
\usepackage{float}
\usepackage{amsmath}
\usepackage[linesnumbered,ruled]{algorithm2e}

\newtheorem{theorem}{Theorem}[section]
\newtheorem{prop}[theorem]{Proposition}
\newtheorem{lemma}[theorem]{Lemma}
\newtheorem{cor}[theorem]{Corollary}

\newtheorem{claim}[theorem]{Claim}
\newtheorem{thm}{Theorem}[section]

\newtheorem{constr}{Construction}

\theoremstyle{definition}
\newtheorem{defn}[theorem]{Definition}
\newtheorem{exam}[theorem]{Example}

\newtheorem{rmk}[theorem]{Remark}

\newenvironment{poc}{\begin{proof}[Proof of claim]}{\end{proof}}

\usepackage{todonotes}

\title{Random Reed--Solomon Codes Correcting Permutations, Insertions, and Deletions over Polynomial-Size Alphabets}
\author{Yijun~Zhang, Yubo~Sun, Xiande~Zhang and Gennian~Ge
\thanks{This research was supported by the National Key Research and Development Program of China under Grant 2025YFC3409900, the National Natural Science Foundation of China under Grant 12231014, and Beijing Scholars Program. The research of X. Zhang was supported by the Innovation Program for Quantum Science and Technology under Grant 2021ZD0302902, by NSFC under Grant 12171452 and Grant 12231014, and by the National Key Research and Development Programs of China under Grant 2023YFA1010200 and Grant 2020YFA0713100.}
\thanks{Y. Zhang ({\tt zyjshuxue@mail.ustc.edu.cn}) is with the School of Mathematical Sciences, University of Science and Technology of China, Hefei, Anhui 230026, China.}
\thanks{Yubo Sun ({\tt ybsun@cnu.edu.cn}) is with the Institute of Mathematics and Interdisciplinary Sciences, Xidian University, Xi'an, Shaanxi 710126, China.}
\thanks{X. Zhang ({\tt drzhangx@ustc.edu.cn}) is with the School of Mathematical Sciences, University of Science and Technology of China, Hefei, Anhui 230026, China, and also with the Hefei National Laboratory, University of Science and Technology of China, Hefei, Anhui 230088, China.}
\thanks{G. Ge ({\tt gnge@zju.edu.cn}) is with the School of Mathematical Sciences, Capital Normal University, Beijing 100048, China.}
}

\begin{document}
\maketitle

\begin{abstract}
    We study Reed--Solomon codes against adversarial coordinate permutations followed by insertion-deletion (insdel) errors. It was previously shown by Con (2025) that Reed--Solomon codes can attain the exact half-Singleton bound in this setting, but only over exponentially large alphabets. We prove that, by allowing an additive $\epsilon n$ gap from this bound, the alphabet size can be reduced to polynomial. More precisely, for fixed constants $R,\epsilon\in(0,1)$ satisfying $2R+\epsilon<1$ and $k=Rn$, a random Reed--Solomon code of length $n$ and dimension $k$ over an alphabet of size $n^{O_{R,\epsilon}(1)}$ is, with high probability, robust against arbitrary coordinate permutations followed by up to $(1-\epsilon)n-2k+1$ insdel errors. 

    We also prove a complementary alphabet-size lower bound, showing that positive-rate codes, which are robust against linearly many insdel errors in the permutation-insdel setting, require a polynomially superlinear alphabet.
    
    Finally, for the explicit two-dimensional Reed--Solomon codes constructed by Con et al. (2024) over alphabet size $O(n^3)$, we give an average $O(n)$-time decoder against arbitrary coordinate permutations followed by $n-3$ insdel errors. Previously, an $O(n)$-time decoder for this code was known only for the deletion setting.
\end{abstract}

\begin{IEEEkeywords}
Reed--Solomon code, insertion-deletion code, permutation channel, secret sharing.
\end{IEEEkeywords}

\section{Introduction}

In modern communication and storage systems, error-correcting codes play a fundamental role in ensuring reliable data transmission and storage. Among them, Reed--Solomon (RS) codes have been extensively studied since their introduction~\cite{Reed1960polynomial}. They remain highly relevant in modern distributed storage systems, where Reed--Solomon codes and related erasure codes are used to provide reliability with low storage overhead~\cite{xia2015tale,Rashmi2013solution}. In recent years, motivated in part by the development of DNA storage, insertion-deletion (insdel) errors have attracted considerable attention in both theoretical research and practical applications~\cite{Levenshtein1965binary, Guruswami2021twodeletion, Sun2024twoedit, Brakensiek2018multipledeletion, Haeupler2021synchronization, Cheng2022deterministic}. It is therefore natural to ask how well Reed--Solomon codes perform against insdel errors.

More broadly, this question belongs to the recent effort to understand linear codes in the insdel metric. Recent work has shown that there are nontrivial linear insdel codes and that they can achieve meaningful rate-error tradeoffs~\cite{Abdel2010correcting, Cheng2023efficient, con2022explicit, xie2024new}. At the same time, unlike the Hamming metric, linearity imposes strong limitations in the insdel metric; a central example is the half-Singleton bound for linear insdel codes~\cite{Cheng2023efficient}.

In this direction, Con, Shpilka, and Tamo~\cite{Con2023reed} proved that Reed--Solomon codes over sufficiently large fields can attain the half-Singleton bound against insdel errors. Later, Con, Guo, Li and Zhang~\cite{Con2024random} showed that, by allowing an additive gap from the half-Singleton bound, random Reed--Solomon codes can achieve a much smaller, in fact linear-size alphabet. Other works on Reed--Solomon codes correcting insdel errors include~\cite{Tonien2007construction, DoDuc2021explicit, Liu2021twodimension, Liu2024optimal, Beelen2026reed, Con2024twodimension, Wang2004deletion}. Nevertheless, the precise capabilities of Reed--Solomon codes, and more generally of linear insdel-correcting codes, remain far from fully understood.

Recently, Con~\cite{Con2025anonymous} introduced the permutation-insdel variant of this problem and related it to fully anonymous Shamir secret sharing. In this model, the coordinates of a codeword may first be arbitrarily permuted before insdel errors occur. This captures the reconstruction requirement in fully anonymous secret sharing, where the decoder receives an unordered collection of shares and does not know the identities of the participants holding them.

Recall that in Shamir's secret-sharing scheme~\cite{Shamir1979how}, the secret is encoded as the constant term of a random low-degree polynomial and shares are evaluations of this polynomial. The fully anonymous version~\cite{Harry2024abuse,Bishop2025fully,Con2025anonymous} requires reconstruction without knowing which evaluation points correspond to the received shares. Using the connection with permutation-insdel Reed--Solomon codes, Con~\cite{Con2025anonymous} constructed fully anonymous $(k-1,2k-1,n)$ ramp schemes, but the construction requires exponentially large alphabets in the constant-rate regime.

\begin{table}[t!]
    \centering
    \small
    \begin{tabular}{|c|c|c|c|}
        \hline
        Setting & Alphabet size & Construction & Reference \\ \hline
        Ordinary insdel, exact & 
        $O\!\left(k^5\left(\frac{en}{k-1}\right)^{4k-4}\right)$ & 
        randomized & 
        \cite{Con2023reed} \\ \hline

        Ordinary insdel, $\epsilon$-gap & 
        $n+2^{\textup{poly}(1/\epsilon)}k$ & 
        randomized & 
        \cite{Con2024random} \\ \hline

        Perm.-insdel, exact & 
        $O\!\left(\left(\binom{n}{2k-1}(2k)!\right)^2\right)$ & 
        randomized & 
        \cite{Con2025anonymous} \\ \hline

        Perm.-insdel, exact & 
        $O\!\left(n^{k^2((2k)!)^2}\right)$ & 
        explicit & 
        \cite{Con2023reed,Con2025anonymous} \\ \hline

        Perm.-insdel, $\epsilon$-gap & 
        $n^{O_{R,\epsilon}(1)}$ & 
        randomized & 
        Theorem~\ref{thm:main informal} \\ \hline
    \end{tabular}
    \caption{Comparison of Reed--Solomon codes for ordinary insdel errors and permutation-insdel errors. Here ``exact'' refers to correcting $n-2k+1$ insdel errors, while ``$\epsilon$-gap'' refers to correcting $(1-\epsilon)n-2k+1$ insdel errors.}
    \label{tab:comparison}
\end{table}

In this paper, we show that the alphabet-size tradeoff from the ordinary insdel setting persists in the permutation-insdel setting. For fixed constants $R,\epsilon\in(0,1)$ with $2R+\epsilon<1$ and $k=Rn$, random Reed--Solomon codes over polynomial-size alphabets are, with high probability, robust against arbitrary coordinate permutations followed by $(1-\epsilon)n-2k+1$ insdel errors. Through the connection established in~\cite{Con2025anonymous}, this also yields fully anonymous $(k-1,2k-1+\epsilon n,n)$ ramp secret-sharing schemes over polynomial-size alphabets.

The comparison with previous work is summarized in Table~\ref{tab:comparison}. It highlights the main regime addressed in this paper: permutation-insdel errors with an additive $\epsilon n$ gap from the half-Singleton bound and polynomial-size alphabets.

We complement this achievability result with a general alphabet-size lower bound in the permutation-insdel setting. It shows that any positive-rate code robust against linearly many insdel errors requires a polynomially superlinear alphabet; in particular, constant-rate Reed--Solomon codes in this model cannot have linear-size alphabets. This is in sharp contrast with the ordinary insdel setting, where random Reed--Solomon codes can approach the half-Singleton bound over linear-size alphabets~\cite{Con2024random}. Thus, the permutation-insdel model exhibits an additional alphabet-size barrier.

Finally, we study efficient decoding in the special two-dimensional case. For $k=2$, Con, Shpilka, and Tamo~\cite{Con2023reed} proved an alphabet-size lower bound for Reed--Solomon codes correcting insdel errors, and this bound was later achieved by the explicit construction of~\cite{Con2024twodimension}. For this construction, Singhvi~\cite{Singhvi2026twodimension} gave an optimal linear-time decoder in the deletion-only setting. We extend this algorithmic picture to the permutation-insdel model: for the explicit two-dimensional Reed--Solomon codes of~\cite{Con2024twodimension}, which have alphabet size $O(n^3)$, we construct an average $O(n)$-time decoder against arbitrary coordinate permutations followed by $n-3$ insdel errors. Our decoding algorithm uses the constant-time inverse-ratio subroutine from the deletion decoder of~\cite{Singhvi2026twodimension}. We also mention the related work~\cite{Banerjee2025decoding}, which gives efficient decoding algorithms when $tk=O(n)$, where $t$ is the number of insdel errors; however, this parameter regime does not reach the half-Singleton bound. For general $k$-dimensional Reed--Solomon codes, determining the optimal alphabet size and designing efficient decoding algorithms remain open problems.

\subsection{Our results}

We now state our main results. Let $\boldsymbol{\alpha}=(\alpha_1,\dots,\alpha_n)$ be chosen uniformly at random from the set of all $n$-tuples of distinct elements in $\mathbb{F}_q$. We call the corresponding code $\textup{RS}_{n,k}(\boldsymbol{\alpha})$ a random Reed--Solomon code over $\mathbb{F}_q$.

Recall that the half-Singleton bound for linear insdel codes gives the upper limit $n-2k+1$ on the number of correctable adversarial insdel errors. Thus, the target radius $(1-\epsilon)n-2k+1$ is an additive $\epsilon n$ gap from the half-Singleton bound.

\begin{thm}[Details in Theorem~\ref{thm:main formal}]\label{thm:main informal}
    Fix constants $R,\epsilon\in(0,1)$ such that $2R+\epsilon<1$, and let $k=Rn$. If $\operatorname{char}(\mathbb F_q)>2k-1+\epsilon n$ and $q\ge n^{O_{R,\epsilon}(1)}$, then a random $\mathrm{RS}_{n,k}(\boldsymbol{\alpha})$ code over $\mathbb{F}_q$ is, with probability at least $1-2^{-n+1}$, robust against the $((1-\epsilon)n-2k+1)$-permutation-insdel adversary.
\end{thm}

For constant rate $R$, Theorem~\ref{thm:main informal} shows that allowing an additive $\epsilon n$ gap from the half-Singleton bound reduces the required alphabet size in the permutation-insdel setting from exponential to polynomial. Through the connection established in~\cite{Con2025anonymous}, it also yields fully anonymous $(k-1,2k-1+\epsilon n,n)$ ramp secret-sharing schemes over polynomial-size alphabets.

We also prove a complementary alphabet-size lower bound in the permutation-insdel setting. Unlike the achievability result above, this lower bound applies to general $q$-ary codes, not necessarily linear. The key observation is that, once arbitrary coordinate permutations are allowed, a codeword is naturally represented by its histogram, namely the multiplicities of its symbols. Under this representation, the histograms of exact deletion outputs are points in a discrete simplex, and pairwise disjointness of the corresponding deletion balls yields the following lower bound.

\begin{thm}\label{thm:alphabet lower informal}
    Let $R,\delta\in(0,1)$ be constants satisfying $R+\delta< 1$. Let $\mathcal{C}\subseteq \Sigma^n$ be a $q$-ary code with $|\mathcal{C}|=q^k$ and $k=Rn$. Suppose that $\mathcal{C}$ is robust against the $t$-permutation-insdel adversary, where $t=\delta n$. Then, for all sufficiently large $n$, the alphabet size satisfies
    $$
    q\ge C_{R,\delta}\, n^{\frac{1-\delta}{1-\delta-R}},
    $$
    where $C_{R,\delta}>0$ depends only on $R$ and $\delta$.
\end{thm}

Finally, for the explicit two-dimensional Reed--Solomon codes of~\cite{Con2024twodimension} over alphabet size $O(n^3)$, we give an average linear-time decoder in the permutation-insdel model. The decoder uses the constant-time inverse-ratio subroutine from the deletion decoder of~\cite{Singhvi2026twodimension}. Our algorithm follows a list-decoding-type approach: we apply the inverse-ratio subroutine to carefully selected length-$3$ subsequences of the received word, producing a list of candidate affine polynomials. A local majority-voting step keeps this list constant in size, and then a hash-table-based overlap test identifies the transmitted codeword in average $O(n)$ time.

\begin{thm}\label{thm:2dim decode informal}
    Let $\mathcal{C}=\textup{RS}_{n,2}(\boldsymbol{\alpha})$ be the code over
    $\mathbb{F}_{q^3}$ specified in Construction~\ref{constr:2dim}. Then there exists a decoder that recovers the transmitted codeword from any received word obtained by applying an arbitrary coordinate permutation followed by at most $n-3$ insdel errors in average $O(n)$ time.
\end{thm}

\subsection{Organization}
In Section~\ref{sec:preliminaries}, we introduce the necessary definitions and recall the relevant frameworks from~\cite{Con2025anonymous,Con2024random}. In Sections~\ref{section:V matrix} and~\ref{section:A matrix}, we analyze the regimes of $\mathbf{V}$-matrix and $\mathbf{A}$-matrix, respectively. In Section~\ref{sec:put it together}, we combine the two estimates and prove Theorem~\ref{thm:main informal}. In Section~\ref{sec:alphabet lower bound}, we prove the lower bound of the alphabet size. Finally, in Section~\ref{sec:decoding}, after recalling the two-dimensional construction and the deletion decoder of~\cite{Singhvi2026twodimension}, we give an average linear-time decoder for the permutation-insdel model.

\section{Preliminaries}\label{sec:preliminaries}

\subsection{Notation and Reed--Solomon codes}

For a positive integer $n$, define $[n]=\{1,2,\dots,n\}$. For integers $a\le b$, define $[a,b]=\{a,a+1,\dots,b\}$. Let $\mathbf{c}=(c_1,c_2,\dots,c_n)$ be a sequence. For any integer $1\le \ell\le n$ and any sequence $I=(i_1,i_2,\dots,i_{\ell})\in [n]^{\ell}$, define $\mathbf{c}_{I}:=(c_{i_1},c_{i_2},\dots,c_{i_{\ell}}).$ Similarly, for the index sequence $I\in [n]^{\ell}$ and any $A=\{a_1,\dots,a_{\ell '}\}\subset [\ell]$ satisfying $a_1<\dots<a_{\ell '}$, define $I_A:=(i_{a_1},\dots,i_{a_{\ell '}})$ and $\mathbf{c}_{I_A}:=(c_{i_{a_1}},c_{i_{a_2}},\dots,c_{i_{a_{\ell '}}})$.

A linear code over a finite field $\mathbb{F}_q$ is a linear subspace $\mathcal{C}\subseteq \mathbb{F}_q^n$. The rate of a linear code $\mathcal{C}$ of block length $n$ and dimension $k=\textup{dim}(\mathcal{C})$ is $R=k/n$.

\begin{defn}[Reed--Solomon code]
    Let $\alpha_1,\alpha_2,\dots,\alpha_n$ be distinct elements of $\mathbb{F}_q$. For $k<n$, the $[n,k]_q$ Reed--Solomon code of dimension $k$ and block length $n$ associated with the evaluation vector $\boldsymbol{\alpha}=(\alpha_1,\alpha_2,\dots,\alpha_n)\in \mathbb{F}_q^n$ is defined as
    $$
    \textup{RS}_{n,k}(\boldsymbol{\alpha})
    :=
    \{(f(\alpha_1),f(\alpha_2),\dots,f(\alpha_n)):
    f\in\mathbb{F}_q[x],\ \textup{deg}(f)<k\}.
    $$
\end{defn}

That is, a codeword of an $[n,k]_q$ Reed--Solomon code is the evaluation vector of a polynomial of degree less than $k$ at $n$ prescribed distinct points of the finite field. For any $f\in\mathbb{F}_q[x]$, we use the shorthand $f(\boldsymbol{\alpha})=(f(\alpha_1),\dots,f(\alpha_n))$.

For later use, we introduce notation for sequentially assigning values to the variables of a symbolic matrix. For simplicity, throughout the paper, we write $\mathbf{X}=(X_1,\ldots,X_n)$ for the tuple of indeterminates. Accordingly, we denote the polynomial ring $\mathbb{F}_q[X_1,\ldots,X_n]$ simply by $\mathbb{F}_q[\mathbf{X}]$. 

\begin{defn}[Partial assignment]
    Let $\mathbf{M}$ be a matrix over $\mathbb{F}_q[\mathbf{X}]$. For $i\in[n]$ and $\alpha_1,\dots,\alpha_i\in\mathbb{F}_q$, denote by
    $$
    \mathbf{M}|_{X_1=\alpha_1,\dots,X_i=\alpha_i}
    $$
    the matrix obtained from $\mathbf{M}$ by substituting $\alpha_j$ for $X_j$ for every $j\in[i]$.
\end{defn}

We now state the Schwarz--Zippel--DeMillo--Lipton Lemma (see, e.g., \cite{Demillo1977probabilistic, Zippel1979probabilistic, Schwartz1980fast}). Here, for a random event $E$, $\textup{Pr}(E)$ denotes the probability that $E$ occurs.

\begin{lemma}[Schwarz--Zippel--DeMillo--Lipton Lemma]\label{lemma:Schwarz-Zippel}
    Let $Q(\mathbf{X})\in \mathbb{F}_q[\mathbf{X}]$ be a nonzero polynomial of total degree $d$, and let $T\subseteq \mathbb{F}_q$. If $\alpha_1,\dots,\alpha_n$ are chosen independently and uniformly at random from $T$, then
    $$
    \textup{Pr}[Q(\alpha_1,\dots,\alpha_n)=0]\le \frac{d}{|T|}.
    $$
\end{lemma}

We also need the following variation of the Schwarz--Zippel--DeMillo--Lipton Lemma from~\cite{Con2024random}, where the evaluation points are sampled without replacement. Here, $\textup{deg}_{X_i}(Q)$ denotes the degree of polynomial $Q$ in the variable $X_i$, where the other variables are viewed as constants.

\begin{lemma}\cite{Con2024random}\label{lemma:Schwarz-Zippel variation}
   Let $Q(\mathbf{X})\in \mathbb{F}_q[\mathbf{X}]$ be a nonzero polynomial such that $\textup{deg}_{X_i}(Q)\le d$ for every $i\in[n]$. Let $T\subseteq \mathbb{F}_q$ be a set of size at least $n$, and let $\boldsymbol{\alpha}=(\alpha_1,\dots,\alpha_n)$ be chosen uniformly from the set of $n$-tuples of distinct elements in $T$. Then
   $$
   \textup{Pr}[Q(\alpha_1,\dots,\alpha_n)=0]\le \frac{nd}{|T|-n+1}.
   $$
\end{lemma}

\subsection{Correcting Insertions, Deletions, and Permutations}

We use the standard notions of subsequences, longest common subsequences, and insertion-deletion edit distance. Let $\mathbf{c}$ be a word over an alphabet $\Sigma$. A deletion removes one symbol from $\mathbf{c}$, while an insertion adds one symbol from $\Sigma$ into an arbitrary position of $\mathbf{c}$. A subsequence of $\mathbf{c}$ is obtained by deleting zero or more symbols. For two words $\mathbf{c}$ and $\mathbf{c}'$, let $\mathrm{LCS}(\mathbf{c},\mathbf{c}')$ denote the length of their longest common subsequence, and let $\mathrm{ED}(\mathbf{c},\mathbf{c}')$ denote the minimum number of insdel errors needed to transform $\mathbf{c}$ into $\mathbf{c}'$. Recall that 
$$
\mathrm{ED}(\mathbf{c},\mathbf{c}')=|\mathbf{c}|+|\mathbf{c}'|-2\mathrm{LCS}(\mathbf{c},\mathbf{c}').
$$
Thus, a code $\mathcal{C}\subseteq\Sigma^n$ corrects $t$ insdel errors if and only if
$\mathrm{LCS}(\mathbf{c},\mathbf{c}')<n-t$
for every pair of distinct codewords $\mathbf{c},\mathbf{c}'\in\mathcal{C}$.

We now turn to the permutation-insdel setting.

\begin{defn}
    A $t$-permutation-insdel adversary is an adversary that, given a codeword $\mathbf{c}=(c_1,\dots,c_n)\in\mathcal{C}\subseteq\Sigma^n$, first chooses an arbitrary permutation $\pi$ of $[n]$ and replaces $\mathbf{c}$ by $\pi(\mathbf{c})=(c_{\pi(1)},\dots,c_{\pi(n)})$, and then performs at most $t$ insdel errors.
\end{defn}

A sequence $I=(i_1,\dots,i_{\ell})\in[n]^\ell$ is called a \emph{distinct-element sequence} if its entries are pairwise distinct. We occasionally identify $I$ with its underlying set $\{i_1,\dots,i_{\ell}\}$ whenever no confusion can arise. 

As in the ordinary insdel setting, robustness can be described by excluding common deletion outputs. The difference is that, after an arbitrary coordinate permutation, the relevant common outputs are indexed not by increasing subsequences, but by arbitrary distinct-element sequences.

\begin{defn}\label{def:equivalent perm insdel}
    A code $\mathcal{C}$ is said to be \emph{robust against the $t$-permutation-insdel adversary} if this adversary cannot produce the same output from two distinct codewords.

    Equivalently, $\mathcal{C}$ is robust against the $t$-permutation-insdel adversary if and only if, for every pair of distinct codewords $\mathbf{c},\mathbf{c}'\in\mathcal{C}$ and every pair of distinct-element sequences $I,J\in[n]^{n-t}$, we have $\mathbf{c}_I\neq \mathbf{c}'_J$.
\end{defn}

\subsection{Con's algebraic rank framework}
We recall the algebraic framework of Con~\cite{Con2025anonymous}. This framework translates permutation-insdel robustness of Reed--Solomon codes into rank conditions for two families of polynomial matrices, the $\mathbf{V}$-matrices and the $\mathbf{A}$-matrices. Most of the notions and rank criteria in this subsection are due to Con~\cite{Con2025anonymous}; we restate them in the notation used in this paper. There are two minor differences in our presentation. First, we use a slightly different, unnormalized convention for $\mathbf{A}$-matrices, which matches the aggregate relations obtained by summing over minimal equal sub-pairs. Second, we record Lemma~\ref{lemma:Vmatrix square converse}, which complements Con's square $\mathbf{V}$-matrix criterion.

By Definition~\ref{def:equivalent perm insdel}, we need to rule out equalities
$\mathbf{c}_I=\mathbf{c}'_J$ for distinct-element sequences $I$ and $J$. For RS codewords, such equalities give a homogeneous linear system in the coefficients of the two message polynomials. The left block of the coefficient matrix records the evaluations on the indexed positions $I$, while the right block records the evaluations on the indexed positions $J$. Since the two constant columns can be combined, this leads to the following $\mathbf{V}$-matrix~\cite{Con2023reed,Con2025anonymous}.

\begin{defn}[$\mathbf{V}$-matrix]
    For positive integers $\ell$ and $k$, and distinct-element sequences
    $I=(i_1,\dots,i_{\ell}),J=(j_1,\dots,j_\ell)\in[n]^\ell$,
    define the $\ell\times(2k-1)$ matrix
    $$
    \mathbf{V}_{k,\ell,I,J}(\mathbf{X})
    :=
    \begin{pmatrix}
    1 & X_{i_1} & \cdots & X_{i_1}^{k-1} & X_{j_1} & \cdots & X_{j_1}^{k-1} \\
    1 & X_{i_2} & \cdots & X_{i_2}^{k-1} & X_{j_2} & \cdots & X_{j_2}^{k-1} \\
    \vdots & \vdots & \ddots & \vdots & \vdots & \ddots & \vdots \\
    1 & X_{i_{\ell}} & \cdots & X_{i_{\ell}}^{k-1} & X_{j_\ell} & \cdots & X_{j_\ell}^{k-1}
    \end{pmatrix}
    $$
    over $\mathbb{F}_q[\mathbf{X}]$. We call
    $\mathbf{V}_{k,\ell,I,J}(\mathbf{X})$
    a $\mathbf{V}$-matrix.
\end{defn}

The following lemma, proved in~\cite{Con2025anonymous}, connects the rank properties of the evaluated $\mathbf{V}$-matrices with the permutation-insdel robustness of the corresponding RS code. Indeed, let
$f(x)=\sum_{i=0}^{k-1}f_ix^i$ and
$g(x)=\sum_{i=0}^{k-1}g_ix^i$ be the message polynomials corresponding to codewords $\mathbf{c}$ and $\mathbf{c}'$, respectively. Then $\mathbf{c}_I=\mathbf{c}'_J$ implies
$$
\mathbf{V}_{k,\ell,I,J}(\mathbf{X})|_{\mathbf{X}=\boldsymbol{\alpha}}
\cdot
(f_0-g_0,f_1,\dots,f_{k-1},-g_1,\dots,-g_{k-1})^{\top}
=0.
$$
The exceptional kernel vectors in Item~1 below correspond only to the case where the two message polynomials are identical and hence do not violate robustness.

\begin{lemma}\cite[Proposition 3]{Con2025anonymous}\label{lemma:RS equivalant condition}
    Let $2k-1\le \ell\le n$ be a positive integer. Let $\boldsymbol{\alpha}=(\alpha_1,\dots,\alpha_n)\in\mathbb{F}_q^n$ have pairwise distinct coordinates. Then the following are equivalent:
    \begin{itemize}
        \item[1.] For every pair of distinct-element sequences $I,J\in[n]^\ell$, either
        $\mathbf{V}_{k,\ell,I,J}(\mathbf{X})|_{\mathbf{X}=\boldsymbol{\alpha}}$
        has full column rank, or every vector in its right kernel is of the form
        $$
        (0,f_1,\dots,f_{k-1},-f_1,\dots,-f_{k-1}).
        $$
        \item[2.] The code $\textup{RS}_{n,k}(\boldsymbol{\alpha})$ is robust against the $(n-\ell)$-permutation-insdel adversary.
    \end{itemize}
\end{lemma}

\begin{rmk}
    In the ordinary insdel setting, the index sequences $I$ and $J$ are increasing. In that case, $\mathbf{V}$-matrices have a relatively simple full-rank criterion: by~\cite[Proposition 2.4]{Con2023reed}, $\mathbf{V}_{k,\ell,I,J}(\mathbf{X})$ has full column rank whenever $I$ and $J$ agree on at most $k-1$ coordinates. In the permutation-insdel setting, the index sequences are arbitrary distinct-element sequences, and the same simple criterion no longer applies.
\end{rmk}

By Lemma~\ref{lemma:RS equivalant condition}, a sufficient approach to proving permutation-insdel robustness would be to choose $\boldsymbol{\alpha}$ so that every evaluated matrix
$\mathbf{V}_{k,\ell,I,J}(\mathbf{X})|_{\mathbf{X}=\boldsymbol{\alpha}}$
has full column rank. However, as observed in~\cite[Remark 6]{Con2025anonymous}, in the permutation-insdel setting the symbolic matrix
$\mathbf{V}_{k,\ell,I,J}(\mathbf{X})$
can fail to have full column rank before evaluation for structural reasons. Therefore, a finer analysis of the combinatorial structure of the pair $(I,J)$ is required. This motivates the separation into the full-rank $\mathbf{V}$-matrix regime and the complementary non-full-rank regime. We begin with the following definitions used to describe this structure.

\begin{defn}
    Let $(I,J)$ be a pair of distinct-element sequences of the same length $\ell$, and let $A\subseteq[\ell]$ be a nonempty set of indices. We call $(I_A,J_A)$ a \emph{minimal equal sub-pair} if
    \begin{itemize}
        \item[1.] $I_A=J_A$ as sets;
        \item[2.] for every nonempty proper subset $A_0\subsetneq A$, we have $I_{A_0}\neq J_{A_0}$ as sets.
    \end{itemize}
    Moreover, a pair $(I,J)$ is \emph{free of equalities} if they have no minimal equal sub-pair.
\end{defn}

\begin{exam}\label{exam:distinct-element sequence}
    Let $I=(1,2,3,4,6,8,9)$ and $J=(2,3,1,6,4,9,7)$. For $A_1=\{1,2,3\}$, the pair $(I_{A_1},J_{A_1})$ is a minimal equal sub-pair. Similarly, for $A_2=\{4,5\}$, the pair $(I_{A_2},J_{A_2})$ is also a minimal equal sub-pair. On the other hand, for $B=\{6,7\}$, we have $I_{B}=(8,9)$ and $J_{B}=(9,7)$, and hence $(I_{B},J_{B})$ is free of equalities, even though $I_{B}$ and $J_{B}$ have a common element.
\end{exam}

The preceding notions lead to the following decomposition of a pair $(I,J)$ into minimal equal sub-pairs and a remaining part that is free of equalities.

\begin{defn}\label{def:decm}
    Let $(I,J)$ be a pair of distinct-element sequences of length $\ell$. We call a partition $A_1,\dots,A_t,B$ of $[\ell]$ a \emph{decomposition} of $(I,J)$ if
    \begin{itemize}
        \item[1.] $(I_{A_j},J_{A_j})$ is a minimal equal sub-pair for every $j\in[t]$;
        \item[2.] $(I_B,J_B)$ is free of equalities.
    \end{itemize}
\end{defn}
Note that we allow $t=0$, in this case the decomposition is simply $B=[\ell]$.

\begin{lemma}\cite[Claim 3]{Con2025anonymous}\label{lemma:decomposition unique}
    Let $(I,J)$ be a pair of distinct-element sequences of length $\ell$. Then, up to reordering the sets, there exists a unique decomposition of $(I,J)$.
\end{lemma}

The decomposition parameter $t$ in Definition~\ref{def:decm} actually controls the rank of square $\mathbf{V}$-matrices.

\begin{lemma}\cite[Proposition 4]{Con2025anonymous}\label{lemma:Vmatrix square full rank}
    Let $\ell=2k-1$. Let $I,J\in[n]^{2k-1}$ be two distinct-element sequences. Assume that $A_1,\dots,A_t,B$ is a decomposition of $(I,J)$. If $t\le k$, then $\textup{det}(\mathbf{V}_{k,2k-1,I,J}(\mathbf{X}))\neq 0$.
\end{lemma}

Thus, $t\le k$ identifies the full-rank $\mathbf{V}$-matrix regime. In this regime, it suffices to show that the symbolic full-rank condition survives the random assignment $\mathbf{X}=\boldsymbol{\alpha}$ with high probability.

We also record the following converse, which complements Lemma~\ref{lemma:Vmatrix square full rank} and will be used to characterize rectangular $\mathbf{V}$-matrices in Section~\ref{section:V matrix}.

\begin{lemma}\label{lemma:Vmatrix square converse}
    Let $I,J\in[n]^{2k-1}$ be two distinct-element sequences.
    Assume that $A_1,\dots,A_t,B$ is the decomposition of the pair $(I,J)$.
    If $t\ge k+1$, then
    $\det\big(\mathbf{V}_{k,2k-1,I,J}(\mathbf{X})\big)=0$
    as a polynomial.
\end{lemma}

\begin{proof}
    We perform column operations on $\mathbf{V}_{k,2k-1,I,J}(\mathbf{X})$ and show that its determinant vanishes. Define a matrix
    $\mathbf{M}=(a_{ij})\in\mathbb{F}_q^{(2k-1)\times(2k-1)}$ by
    $$
    a_{ij}=
    \begin{cases}
        1, & \text{if } i=j,\\
        -1, & \text{if } i+k-1=j \text{ and } i\ge 2,\\
        0, & \text{otherwise}.
    \end{cases}
    $$
    Set $\mathbf{V}'=\mathbf{V}_{k,2k-1,I,J}(\mathbf{X})\cdot \mathbf{M}.$ Explicitly,
    $$
    \mathbf{V}'=
    \begin{pmatrix}
        1 & X_{i_1} & \cdots & X_{i_1}^{k-1} & X_{j_1}-X_{i_1} & \cdots & X_{j_1}^{k-1}-X_{i_1}^{k-1} \\
        1 & X_{i_2} & \cdots & X_{i_2}^{k-1} & X_{j_2}-X_{i_2} & \cdots & X_{j_2}^{k-1}-X_{i_2}^{k-1} \\
        \vdots & \vdots & \ddots & \vdots & \vdots & \ddots & \vdots \\
        1 & X_{i_{2k-1}} & \cdots & X_{i_{2k-1}}^{k-1} & X_{j_{2k-1}}-X_{i_{2k-1}} & \cdots & X_{j_{2k-1}}^{k-1}-X_{i_{2k-1}}^{k-1}
    \end{pmatrix}.
    $$
    
    \begin{claim}
        Every $(k-1)\times(k-1)$ submatrix of $\mathbf{V}'$ obtained by taking the last $k-1$ columns and any $k-1$ rows has a zero determinant.
    \end{claim}
    
    \begin{poc}
        Let $\mathbf{W}$ be such a $(k-1)\times(k-1)$ submatrix. We first show that the rows of $\mathbf{W}$ must contain all rows corresponding to some block $A_i$, where $i\in[t]$.
        
         Suppose, for contradiction, that for every $i\in[t]$, the rows of $\mathbf{W}$ avoid at least one row indexed by $A_i$. Then $\mathbf{W}$ can contain at most $|A_i|-1$ rows from each $A_i$. Therefore, the total number of rows that $\mathbf{W}$ can contain is at most
        $$
        \sum_{i=1}^{t}(|A_i|-1)+|B|
        =
        \left(\sum_{i=1}^{t}|A_i|+|B|\right)-t
        \le
        (2k-1)-(k+1)
        =
        k-2,
        $$
        which contradicts the fact that $\mathbf{W}$ has $k-1$ rows.
        
        Hence, there exists some $i\in[t]$ such that $\mathbf{W}$ contains all rows indexed by $A_i$. Since $I_{A_i}=J_{A_i}$ as sets, we have
        $$
        \sum_{s\in A_i}\left(X_{j_s}^{m}-X_{i_s}^{m}\right)=0
        \quad
        \text{for every } m\in[k-1].
        $$
        Thus the rows of $\mathbf{W}$ are linearly dependent, which completes the proof of claim.
    \end{poc}
    
    Finally, since $\mathbf{M}$ is invertible, we have
    $$
    \det(\mathbf{V}')=\det\big(\mathbf{V}_{k,2k-1,I,J}(\mathbf{X})\big)\det(\mathbf{M}).
    $$
    It remains to show that $\det(\mathbf{V}')=0$. Expanding $\det(\mathbf{V}')$ along the last $k-1$ columns, every term contains the determinant of a $(k-1)\times(k-1)$ submatrix formed from these columns. By the claim, all such determinants are zero, so $\det(\mathbf{V}')=0$, which completes the proof.
\end{proof}

The converse above shows that, when $t\ge k+1$, the square $\mathbf{V}$-matrix is singular as a polynomial. In this regime, by Lemma~\ref{lemma:RS equivalant condition}, the full-rank $\mathbf{V}$-matrix strategy is no longer available, and one has to use the additional structure coming from the minimal equal sub-pairs. If two RS codewords agree on the indexed positions specified by $(I,J)$, then each block $A_j$ with $I_{A_j}=J_{A_j}$ allows us to sum the corresponding coordinate equalities over a common set of evaluation points. The resulting aggregate relations are governed by power sums, which leads to the following variant of the $\mathbf{A}$-matrices used by Con~\cite{Con2025anonymous}.

\begin{defn}[$\mathbf{A}$-matrix]
    For pairwise disjoint nonempty sets $I^1,\dots,I^t\subseteq[n]$, define the matrix
    $$
    \mathbf{A}_{k,I^1,\dots,I^t}(\mathbf{X})
    :=
    \begin{pmatrix}
    |I^1| & \sum_{j\in I^1}X_j & \sum_{j\in I^1}X_j^2 & \cdots & \sum_{j\in I^1}X_j^{k-1} \\
    |I^2| & \sum_{j\in I^2}X_j & \sum_{j\in I^2}X_j^2 & \cdots & \sum_{j\in I^2}X_j^{k-1} \\
    \vdots & \vdots & \vdots & \ddots & \vdots \\
    |I^t| & \sum_{j\in I^t}X_j & \sum_{j\in I^t}X_j^2 & \cdots & \sum_{j\in I^t}X_j^{k-1}
    \end{pmatrix},
    $$
    with $\mathbf{X}=(X_1,\dots,X_n)$. We call $\mathbf{A}_{k,I^1,\dots,I^t}(\mathbf{X})$ an $\mathbf{A}$-matrix.
\end{defn}

\begin{rmk}
    The $\mathbf{A}$-matrix used in~\cite{Con2025anonymous} has first column equal to $1$. However, in the summation argument over a minimal equal sub-pair, the coefficient of the constant term should be the size of the corresponding set. Thus, for the aggregate relations used in the rank criterion, the natural matrix has first column entries $|I^1|,\dots,|I^t|$. We therefore use this unnormalized version of the $\mathbf{A}$-matrix. Under the assumption $\operatorname{char}(\mathbb{F}_q)>\ell$, all set sizes that occur in our application are nonzero in $\mathbb{F}_q$, and the full-rank arguments can be carried out with this modified form.
    
\end{rmk}

The following lemma records a robust full-rank property of the unnormalized $\mathbf{A}$-matrices used in this paper.

\begin{lemma}\label{lemma:Amatrix full rank}
    Let $I^1,\dots,I^t\subseteq[n]$ be pairwise disjoint nonempty sets with
    $t\ge k$ and $\sum_{s=1}^t |I^s|\le \ell$. Suppose that
    $\operatorname{char}(\mathbb{F}_q)>\ell$. Then every $k\times k$ submatrix of
    $\mathbf{A}_{k,I^1,\dots,I^t}(\mathbf{X})$ has full rank. Moreover, if $\mathbf{M}$ is any such $k\times k$ submatrix, then
    $\deg_{X_i}(\det(\mathbf{M}))\le k-1$
    for every $i\in[n]$.
\end{lemma}

\begin{proof}
    Consider any $k\times k$ submatrix obtained by choosing the rows corresponding to
    $I^{s_1},\dots,I^{s_k}$. For each $a\in[k]$, set all variables $X_j$ with $j\in I^{s_a}$ equal to a new variable $Y_a$. Since the sets $I^{s_1},\dots,I^{s_k}$ are pairwise disjoint, this specialization is consistent. The $a$-th selected row becomes
    $$
    |I^{s_a}|(1,Y_a,Y_a^2,\dots,Y_a^{k-1}).
    $$
    Since $1\le |I^{s_a}|\le \ell$ and $\operatorname{char}(\mathbb{F}_q)>\ell$, each factor $|I^{s_a}|$ is nonzero in $\mathbb{F}_q$. Thus the specialized submatrix is a nonzero diagonal scaling of a Vandermonde matrix, and hence has nonzero determinant. Therefore the original determinant is not the zero polynomial.

    The degree bound follows because the sets $I^1,\dots,I^t$ are pairwise disjoint: each variable $X_i$ appears in at most one row of $\mathbf{M}$, and in that row with degree at most $k-1$.
\end{proof}

The full-rank property of $\mathbf{A}$-matrices rules out nontrivial polynomial relations satisfying the aggregate constraints coming from many minimal equal sub-pairs. Together with the full-rank $\mathbf{V}$-matrix regime, this gives the following rank criterion. This criterion is adapted from~\cite[Proposition 8]{Con2025anonymous}; here we use the unnormalized $\mathbf{A}$-matrices defined above, so we also impose the condition $\operatorname{char}(\mathbb{F}_q)>\ell$ to ensure that all relevant set sizes are nonzero in $\mathbb{F}_q$.

\begin{lemma}\label{lemma:V A two conditions}
    Let $\rho,\epsilon\in(0,1)$ with $0<\rho<\epsilon$, and let $k\le n$ be positive integers with $k=Rn$. Let $\ell=2k-1+\epsilon n$ and $t=k+\rho n$, and suppose that $2R+\epsilon<1$. Let $\boldsymbol{\alpha}=(\alpha_1,\dots,\alpha_n)\in\mathbb{F}_q^n$ be an $n$-tuple of distinct elements, and assume that $\operatorname{char}(\mathbb{F}_q)>\ell$. Suppose that the following two conditions hold:
    \begin{enumerate}
        \item \label{item:1} For every pair $(I,J)$ of distinct-element sequences of length $\ell$ containing at most $t$ minimal equal sub-pairs, the matrix $\mathbf{V}_{k,\ell,I,J}(\mathbf{X})|_{\mathbf{X}=\boldsymbol{\alpha}}$ has full column rank.

        \item \label{item:2} For every collection of pairwise disjoint nonempty sets $I^1,\dots,I^t\subseteq[n]$ satisfying $\sum_{i=1}^{t}|I^i|\le \ell$, the matrix $\mathbf{A}_{k,I^1,\dots,I^t}(\mathbf{X})|_{\mathbf{X}=\boldsymbol{\alpha}}$ has full column rank.
     \end{enumerate}
     Then the corresponding $\textup{RS}_{n,k}(\boldsymbol{\alpha})$ code is robust against the $(n-\ell)$-permutation-insdel adversary.
\end{lemma}

The proof follows the same argument as~\cite[Proposition 8]{Con2025anonymous}. Briefly, suppose that the pair $(I,J)$ contains at least $t$ minimal equal sub-pairs indexed by $A_1,\dots,A_t$, where $t=k+\rho n$. For each $i\in[t]$, set
$I^i=I_{A_i}=J_{A_i}$ as a subset of $[n]$. Suppose that $f(x)=\sum_{s=0}^{k-1}f_sx^s$ and $g(x)=\sum_{s=0}^{k-1}g_sx^s$ are two polynomials whose evaluations satisfy the corresponding equations encoded by
$\mathbf{V}_{k,\ell,I,J}(\mathbf{X})|_{\mathbf{X}=\boldsymbol{\alpha}}$. Summing these equations over each $A_i$ gives
$$
|I^i|(f_0-g_0)
+
\sum_{s=1}^{k-1}
\left(\sum_{a\in I^i}\alpha_a^s\right)(f_s-g_s)
=0.
$$
Thus the coefficient vector
$(f_0-g_0,f_1-g_1,\dots,f_{k-1}-g_{k-1})^\top$
lies in the right kernel of
$\mathbf{A}_{k,I^1,\dots,I^t}(\mathbf{X})|_{\mathbf{X}=\boldsymbol{\alpha}}$.
By Condition~\ref{item:2}, this $\mathbf{A}$-matrix has full column rank, so $f=g$. Hence the corresponding right-kernel vectors of
$\mathbf{V}_{k,\ell,I,J}(\mathbf{X})|_{\mathbf{X}=\boldsymbol{\alpha}}$
are only of the desired form in Lemma~\ref{lemma:RS equivalant condition}.

\begin{rmk}
    Conditions~\ref{item:1} and~\ref{item:2} in Lemma~\ref{lemma:V A two conditions} overlap when $(I,J)$ contains exactly $k+\rho n$ minimal equal sub-pairs. This overlap is harmless and is kept for notational simplicity.
\end{rmk}

\subsection{The faulty index framework}

We recall the faulty index framework of Con, Guo, Li, and Zhang~\cite{Con2024random}. The algorithm below is adapted from their work. In later sections, we will apply it after specifying deterministic rules for selecting full-rank square submatrices in the $\mathbf{V}$- and $\mathbf{A}$-matrix regimes.

\begin{defn}\label{def:matrix retriction}
    Let $\mathbf{M}\in \mathbb{F}_q[\mathbf{X}]^{\ell\times m}$ be a matrix with $\ell\ge m$. For a subset $\mathcal{R}\subseteq[n]$, define $\mathbf{M}^{\mathcal{R}}$ as the submatrix of $\mathbf{M}$ obtained by deleting every row that contains an entry depending on some variable $X_i$ with $i\in \mathcal{R}$.
\end{defn}

The faulty index framework tracks when a full-rank square polynomial matrix loses rank during sequential assignment $X_1=\alpha_1,\dots,X_n=\alpha_n$. If the determinant is nonzero before assigning $X_i$ but becomes zero after assigning $X_i$, then the $i$-th assignment is the first step at which the rank failure occurs. This motivates the following definition. For $i\in[n]$, write
$\mathbf{X}_{<i}:=(X_1,\dots,X_{i-1})$ and
$\mathbf{X}_{\le i}:=(X_1,\dots,X_i)$; define
$\boldsymbol{\alpha}_{<i}$ and $\boldsymbol{\alpha}_{\le i}$ similarly.

\begin{defn}[Faulty index]\label{def:faulty index}
    Let $\mathbf{M}\in \mathbb{F}_q[\mathbf{X}]^{m\times m}$ be a square matrix with a nonzero determinant. For $\boldsymbol{\alpha}=(\alpha_1,\dots,\alpha_n)\in\mathbb{F}_q^n$, we say that $i\in[n]$ is a \emph{faulty index} of $\mathbf{M}$ with respect to $\boldsymbol{\alpha}$ if $\det\left(\mathbf{M}|_{\mathbf{X}_{<i}=\boldsymbol{\alpha}_{<i}}\right)\neq 0$, but $\det\left(\mathbf{M}|_{\mathbf{X}_{\le i}=\boldsymbol{\alpha}_{\le i}}\right)=0$.
\end{defn}

Algorithm~\ref{alg:faulty index output}
attempts to certify that
$\mathbf{M}|_{\mathbf{X}=\boldsymbol{\alpha}}$
has full column rank. It iteratively deletes rows involving variables $X_i$ whose indices $i$ have already been identified as faulty, selects a full-rank square submatrix of the remaining symbolic matrix, and tests whether this submatrix loses rank during the sequential assignment. The algorithm outputs ``SUCCESS'' if full column rank is certified, outputs ``FAIL'' if the remaining symbolic matrix loses full column rank, and otherwise outputs a faulty index sequence $(u_1,\dots,u_r)$. In our applications, $r$ is chosen so that the ``FAIL'' outcome does not occur; hence any rank failure after evaluation yields a length-$r$ faulty index sequence. Bounding the rank failure probability is then reduced to bounding the probability of producing possible faulty index sequences.

\begin{algorithm}[htbp]
\caption{Certify Full Column Rank}
\label{alg:faulty index output}
\KwIn{$n,m,r,\ell\in \mathbb{Z}^+$ with $\ell\ge m$, a full-column-rank matrix $\mathbf{M}\in \mathbb{F}_q[\mathbf{X}]^{\ell\times m}$, and evaluation points $\alpha_1,\dots,\alpha_n\in\mathbb{F}_q$.}
\KwOut{``SUCCESS'', ``FAIL'', or a faulty index sequence $(u_1,\dots,u_r)\in[n]^r$.}

$\mathcal{R}\leftarrow \varnothing$. Here, $\mathcal{R}$ indicates the set of faulty evaluation indices.\\
\For{$j=1$ \KwTo $r$}{
    \If{$\textup{rank}(\mathbf{M}^{\mathcal{R}})<m$}{
        Output ``FAIL'' and halt.
    }
    Select an $m\times m$ full-rank square submatrix of $\mathbf{M}^{\mathcal{R}}$, denoted by $\mathbf{M}^{\mathcal{R}}_{P}$. \\
    \eIf{there exists a faulty index $i\in[n]$ of $\mathbf{M}^{\mathcal{R}}_{P}$}{
        $u_j\leftarrow i$ and $\mathcal{R}\leftarrow \mathcal{R}\cup\{i\}$.
    }{
        Output ``SUCCESS'' and halt.
    }
}
Output $(u_1,\dots,u_r)$.
\end{algorithm}

\begin{rmk}
    If the choice of the full-rank square submatrix $\mathbf{M}^{\mathcal{R}}_{P}$ is unique for each fixed $\mathcal{R}$, then the output faulty index sequence $(u_1,\dots,u_r)$ is also unique, whenever it exists.
\end{rmk}

\section{Achieving Polynomial Alphabet Size}\label{sec:achieve poly alphabet}

In this section, we prove that random Reed--Solomon codes over polynomial-size alphabets can approach the half-Singleton bound against permutation-insdel errors. The proof combines the algebraic rank framework of~\cite{Con2025anonymous} with the faulty index method of~\cite{Con2024random}. Let $\ell=2k-1+\epsilon n$. By Lemma~\ref{lemma:V A two conditions}, correcting $n-\ell=n-2k+1-\epsilon n$ insdel errors, even after arbitrary coordinate permutations, reduces to verifying the $\mathbf{V}$- and $\mathbf{A}$-matrix rank conditions after random assignment $\mathbf{X}=\boldsymbol{\alpha}$.

We introduce a parameter $\rho$ to separate the two regimes and balance their failure probabilities. For simplicity of presentation, we assume that $\epsilon n$, $\rho n$, and $(\epsilon-\rho)n/2$ are integers; the general case follows by replacing these quantities with suitable floors and ceilings. The next two subsections estimate the failure probabilities in the $\mathbf{V}$-matrix and $\mathbf{A}$-matrix regimes separately.

\subsection{$\mathbf{V}$-Matrix}\label{section:V matrix}

In this section, we analyze the $\mathbf{V}$-matrix regime. We first characterize when the rectangular matrix $\mathbf{V}_{k,\ell,I,J}(\mathbf{X})$ has full column rank and then derive a deterministic rule for selecting a full-rank square submatrix. This selection rule allows us to apply the faulty index framework and to bound the probability that $\mathbf{V}_{k,\ell,I,J}(\mathbf{X})|_{\mathbf{X}=\boldsymbol{\alpha}}$ fails to have full column rank. Before proving the rank characterization, we record how a decomposition of $(I,J)$ induces a decomposition of the restricted pair $(I_P,J_P)$.

\begin{lemma}\label{lemma:decomposition restriction}
    Let $I,J\in[n]^\ell$ be two distinct-element sequences. Assume that $A_1,\dots,A_t,B$ is a decomposition of $(I,J)$. Let $P\subseteq[\ell]$, and $s_1,\dots,s_{t'}\in[t]$ be all indices such that $A_{s_j}\subseteq P$. Then $A_{s_1},\dots,A_{s_{t'}},B'$ is a decomposition of $(I_P,J_P)$, where
    $$
    B'=P\setminus \bigcup_{j=1}^{t'}A_{s_j}
    =
    \bigcup_{i\notin\{s_1,\dots,s_{t'}\}}(A_i\cap P)\cup(B\cap P).
    $$
\end{lemma}

\begin{proof}
    By Lemma~\ref{lemma:decomposition unique}, it suffices to show that $(I_{B'},J_{B'})$ is free of equalities. Suppose, for contradiction, that there exists a nonempty subset $A_0\subseteq B'$ such that $I_{A_0}=J_{A_0}$ as sets. By replacing $A_0$ with a minimal nonempty subset if necessary, we may assume that $(I_{A_0},J_{A_0})$ is a minimal equal sub-pair of $(I_P,J_P)$.

    We claim that $A_0\cap A_i=\emptyset$ for every $i\in[t]$. Suppose otherwise, and fix $i\in[t]$ such that $C:=A_0\cap A_i$ is nonempty. Since $B'$ does not contain an entire block $A_i$, the set $C$ is a proper subset of $A_i$. The minimality of $(I_{A_i},J_{A_i})$ gives $I_C\neq J_C$ as sets. Hence there exists $a\in C$ such that $i_a\notin J_C$.

    On the other hand, $I_{A_0}=J_{A_0}$ and $I_{A_i}=J_{A_i}$ as sets. Since $a\in C$ and $i_a\notin J_C$, the symbol $i_a$ must appear both in $J_{A_0\setminus C}$ and in $J_{A_i\setminus C}$. This contradicts the fact that $J$ is a distinct-element sequence, because $A_0\setminus C$ and $A_i\setminus C$ are disjoint.
    
    Thus $A_0$ is disjoint from every $A_i$, so $A_0\subseteq B$. This contradicts that $(I_B,J_B)$ is free of equalities. Therefore $(I_{B'},J_{B'})$ is free of equalities, and the lemma follows.
\end{proof}

Here, for convenience, the indices of $I_P$ and $J_P$ are inherited from the original sequences $I$ and $J$, rather than renumbered. The following example illustrates Lemma~\ref{lemma:decomposition restriction}.

\begin{exam}
    For $\ell=7$, let $I=(1,2,3,4,6,8,9)$ and $J=(2,3,1,6,4,9,7)$. According to Example~\ref{exam:distinct-element sequence}, $A_1=\{1,2,3\}$, $A_2=\{4,5\}$, and $B=\{6,7\}$ form a decomposition of $(I,J)$. Take $P=\{2,3,4,5,6\}$. Then $I_P=(2,3,4,6,8)$ and $J_P=(3,1,6,4,9)$. Observe that $A_1\cap P=\{2,3\}\subsetneq A_1$, $A_2\subseteq P$, and $B\cap P=\{6\}$. Hence the only minimal equal sub-pair that remains intact is $A_2$, and one can check that $A_2$ together with $B'=P\setminus A_2$ gives the decomposition of $(I_P,J_P)$.
\end{exam}

\begin{cor}\label{cor:decomposition delete}
    Let $I,J\in[n]^\ell$ be two distinct-element sequences. Assume that $A_1,\dots,A_t,B$ is a decomposition of $(I,J)$. Let $P\subseteq[\ell]$ have size $\ell'$. Suppose that $A_1',\dots,A_{t'}',B'$ is the decomposition of $(I_P,J_P)$. Then $t-\ell+\ell'\le t'\le t$. 
\end{cor}

\begin{proof}
    The upper bound $t'\le t$ follows directly from Lemma~\ref{lemma:decomposition restriction}, since the only minimal equal sub-pairs that can remain in the decomposition of $(I_P,J_P)$ are those original blocks $A_i$ that are completely contained in $P$. 

    For the lower bound, note that an original block $A_i$ fails to remain a minimal equal sub-pair of $(I_P,J_P)$ only if $A_i$ intersects the deleted set $[\ell]\setminus P$. Since $|[\ell]\setminus P|=\ell-\ell'$ and the sets $A_1,\dots,A_t$ are pairwise disjoint, at most $\ell-\ell'$ of the original blocks can be destroyed. This implies that at least $t-\ell+\ell'$ minimal equal sub-pairs remain, and thus $t'\ge t-\ell+\ell'$.
\end{proof}

The following lemma gives an equivalent condition for $\mathbf{V}_{k,\ell,I,J}(\mathbf{X})$ to have full column rank. It generalizes Lemma~\ref{lemma:Vmatrix square full rank} and Lemma~\ref{lemma:Vmatrix square converse} from the square case to the rectangular case. For any
$P\subseteq[\ell]$ of size $\ell'$, we denote by $\mathbf{V}_{k,\ell',I_P,J_P}(\mathbf{X})$ the submatrix of $\mathbf{V}_{k,\ell,I,J}(\mathbf{X})$ consisting of the rows indexed by $P$.

\begin{lemma}\label{lemma:Vmatrix ell equivalent}
    Let $I,J\in[n]^\ell$ be two distinct-element sequences. Assume that $A_1,\dots,A_t,B$ is a decomposition of $(I,J)$. Then $\mathbf{V}_{k,\ell,I,J}(\mathbf{X})$ has full column rank if and only if $t\le \ell-k+1$.
\end{lemma}

\begin{proof}
    First assume that $t\le \ell-k+1$. We will choose a subset $P\subseteq[\ell]$ of size $2k-1$ such that the decomposition of $(I_P,J_P)$ contains at most $k$ minimal equal sub-pairs. This will imply that the square submatrix $\mathbf{V}_{k,2k-1,I_P,J_P}(\mathbf{X})$ has full rank by Lemma~\ref{lemma:Vmatrix square full rank}, and hence $\mathbf{V}_{k,\ell,I,J}(\mathbf{X})$ has full column rank.

    Since we need to select $2k-1$ rows, we delete $\ell-(2k-1)$ indices from $[\ell]$. Let $s=\min\{\ell-2k+1,t\}$. Choose one index from each of $A_1,\dots,A_{s}$, and then, if necessary, choose additional indices arbitrarily until exactly $\ell-2k+1$ indices have been chosen. Let $\overline{P}$ be the set of chosen indices and set $P=[\ell]\setminus\overline{P}$. Then $|P|=2k-1$, and the blocks $A_1,\dots,A_{s}$ are not completely contained in $P$. By Lemma~\ref{lemma:decomposition restriction}, these blocks do not remain minimal equal sub-pairs in the decomposition of $(I_P,J_P)$. This implies that the number of minimal equal sub-pairs in the decomposition of $(I_P,J_P)$ is at most $t-s\le k$. This shows that $\mathbf{V}_{k,\ell,I,J}(\mathbf{X})$ has full column rank.

    Conversely, assume that $t\ge \ell-k+2$. We show that $\mathbf{V}_{k,\ell,I,J}(\mathbf{X})$ does not have full column rank. It suffices to show that every square submatrix formed by choosing $2k-1$ rows is singular. Let $P\subseteq[\ell]$ be any subset of size $2k-1$. Suppose that the decomposition of $(I_P,J_P)$ contains $t'$ minimal equal sub-pairs. By Corollary~\ref{cor:decomposition delete}, we have
    $t'\ge t-\ell+(2k-1)\ge k+1$.
    Therefore, by Lemma~\ref{lemma:Vmatrix square converse}, the square matrix $\mathbf{V}_{k,2k-1,I_P,J_P}(\mathbf{X})$ is singular. Since this holds for every such $P$, the rectangular matrix $\mathbf{V}_{k,\ell,I,J}(\mathbf{X})$ cannot have full column rank.
\end{proof}

The proof of Lemma~\ref{lemma:Vmatrix ell equivalent} yields a deterministic procedure for selecting a full-rank square submatrix from a full-column-rank $\mathbf{V}$-matrix; see Algorithm~\ref{alg:select full rank submatrix}.

\begin{algorithm}[h]
\caption{Select a Full-Rank Square $\mathbf{V}$-Submatrix}
\label{alg:select full rank submatrix}
\KwIn{$n,k,\ell\in \mathbb{Z}^+$ with $\ell\ge 2k-1$, distinct-element sequences $I,J\in[n]^\ell$, a positive integer $t\le \ell-k+1$, and a decomposition $A_1,\dots,A_t,B$ of $(I,J)$.}
\KwOut{A subset $P\subseteq[\ell]$ of size $2k-1$ such that $\mathbf{V}_{k,2k-1,I_P,J_P}(\mathbf{X})$ has full rank.}

$\overline{P}\leftarrow \varnothing$. Here $\overline{P}$ indicates the set of deleted rows.\\
$s\leftarrow \textup{min}\{\ell-2k+1,t\}$. \\
\For{$j\leftarrow 1$ \KwTo $s$}{
    Select the smallest element $a_j\in A_j$. \\
    $\overline{P}\leftarrow \overline{P}\cup\{a_j\}$. 
}
\While{$|\overline{P}|<\ell-2k+1$}{
    Select the smallest element $a\in[\ell]\setminus \overline{P}$. \\
    $\overline{P}\leftarrow \overline{P}\cup\{a\}$. \\
}
$P\leftarrow[\ell]\setminus \overline{P}$.  Here $P$ indicates the set of selected rows.\\ 
Output $P$ and the corresponding matrix $\mathbf{V}_{k,2k-1,I_P,J_P}(\mathbf{X})$.
\end{algorithm}

\begin{rmk}
    It is worth noting that Algorithm~\ref{alg:select full rank submatrix} applies to any $\ell\ge 2k-1$. The strategy of always choosing the smallest available element makes the output subset $P$ uniquely determined by the input data.
\end{rmk}

The following corollary follows directly from Corollary~\ref{cor:decomposition delete} and Lemma~\ref{lemma:Vmatrix ell equivalent}. It shows that deleting rows does not destroy the full column rank as long as enough rows remain.

\begin{cor}\label{cor:Vmatrix delete full rank}
    Let $I,J\in[n]^\ell$ be two distinct-element sequences. Assume that $A_1,\dots,A_t,B$ is a decomposition of $(I,J)$. Let $P\subseteq[\ell]$ be a subset of size $\ell'$. If $t\le \ell'-k+1$, then $\mathbf{V}_{k,\ell',I_P,J_P}(\mathbf{X})$ has full column rank.
\end{cor}

Under the notation in Definition~\ref{def:matrix retriction}, for any $\mathcal{R}\subseteq[n]$, let ${\mathbf{V}}^{\mathcal{R}}_{k,\ell,I,J}(\mathbf{X})$ denote the submatrix of $\mathbf{V}_{k,\ell,I,J}(\mathbf{X})$ obtained by deleting all rows that contain some variable $X_i$ with $i\in \mathcal{R}$. The following corollary shows that the $\mathbf{V}$-matrix remains full column rank after deleting all rows involving a small set of variables.

\begin{cor}\label{cor:V^R matrix has full rank}
    Let $0<\rho<\epsilon$, $\ell=2k-1+\epsilon n$, and $r=(\epsilon-\rho)n/2$. Let $I,J\in[n]^\ell$ be two distinct-element sequences. Suppose that $A_1,\dots,A_t,B$ is a decomposition of $(I,J)$ with $t\le k+\rho n$. Then, for any $\mathcal{R}\subseteq[n]$ with $|\mathcal{R}|\le r$, the matrix ${\mathbf{V}}^{\mathcal{R}}_{k,\ell,I,J}(\mathbf{X})$ has the full column rank.
\end{cor}

\begin{proof}
    Define $P:=[\ell]\setminus\{j\in[\ell]: I_j\in \mathcal{R} \textup{ or } J_j\in \mathcal{R}\}$ and let $\ell'=|P|$. Since $|\mathcal{R}|\le r$ and each index in $\mathcal{R}$ appears at most once in $I$ and once in $J$, at most $2r$ rows are deleted by the definition of ${\mathbf{V}}^{\mathcal{R}}_{k,\ell,I,J}(\mathbf{X})$, and hence
    $$
    \ell'\ge \ell-2r
    =
    2k-1+\epsilon n-(\epsilon-\rho)n
    =
    2k-1+\rho n.
    $$
    Therefore, $t\le k+\rho n\le \ell'-k+1$. By Corollary~\ref{cor:Vmatrix delete full rank}, the matrix ${\mathbf{V}}^{\mathcal{R}}_{k,\ell,I,J}(\mathbf{X})$ has the full column rank.
\end{proof}

Thus, under the conditions of Corollary~\ref{cor:V^R matrix has full rank}, we may apply Algorithm~\ref{alg:select full rank submatrix} to ${\mathbf{V}}^{\mathcal{R}}_{k,\ell,I,J}(\mathbf{X})$. For simplicity, when the input matrix of Algorithm~\ref{alg:select full rank submatrix} is ${\mathbf{V}}^{\mathcal{R}}_{k,\ell,I,J}(\mathbf{X})$, we denote the selected full-rank square submatrix by ${\mathbf{V}}^{\mathcal{R}}_{k,\ell,I,J,P}(\mathbf{X})$, where $P=P(\mathcal{R})\subseteq[\ell]$ is the corresponding set of selected rows.

Next, we combine Algorithm~\ref{alg:faulty index output} with Algorithm~\ref{alg:select full rank submatrix}. Specifically, in Line 6 of Algorithm~\ref{alg:faulty index output}, we apply Algorithm~\ref{alg:select full rank submatrix} to the current matrix ${\mathbf{V}}^{\mathcal{R}}_{k,\ell,I,J}(\mathbf{X})$ to select a full-rank square submatrix. The following lemma describes the possible outputs of the combined procedure when applied to $\mathbf{V}_{k,\ell,I,J}(\mathbf{X})$.

\begin{lemma}\label{lemma:behavior of alg V}
    Let $0<\rho<\epsilon$, $\ell=2k-1+\epsilon n$, and $r=(\epsilon-\rho)n/2$. Let $I,J\in[n]^\ell$ be two distinct-element sequences. Suppose that $A_1,\dots,A_t,B$ is the decomposition of $(I,J)$ with $t\le k+\rho n$. Then, for any pairwise distinct $\alpha_1,\dots,\alpha_n\in\mathbb{F}_q$, running Algorithm~\ref{alg:faulty index output} on the matrix $\mathbf{V}_{k,\ell,I,J}(\mathbf{X})$, while using Algorithm~\ref{alg:select full rank submatrix} on Line 6 to select the square submatrix, yields one of the following two outcomes:
    \begin{itemize}
        \item[1.] Algorithm~\ref{alg:faulty index output} outputs ``SUCCESS''. In this case, $\mathbf{V}_{k,\ell,I,J}(\mathbf{X})|_{\mathbf{X}=\boldsymbol{\alpha}}$ has full column rank.
        \item[2.] Algorithm~\ref{alg:faulty index output} outputs a faulty index sequence $(u_1,\dots,u_r)\in[n]^r$. The indices $u_1,\dots,u_r$ are distinct. Moreover, for each $j\in[r]$, $u_j$ is the faulty index of ${\mathbf{V}}^{\mathcal{R}_j}_{k,\ell,I,J,P_j}(\mathbf{X})$, where $\mathcal{R}_j:=\{u_1,\dots,u_{j-1}\}$ and ${\mathbf{V}}^{\mathcal{R}_j}_{k,\ell,I,J,P_j}(\mathbf{X})$ is the full-rank square submatrix selected by applying Algorithm~\ref{alg:select full rank submatrix} to ${\mathbf{V}}^{\mathcal{R}_j}_{k,\ell,I,J}(\mathbf{X})$.
    \end{itemize}
\end{lemma}

\begin{proof}
    We first show that Algorithm~\ref{alg:faulty index output} never outputs ``FAIL''. Since $|\mathcal{R}_j|\le r$, by Corollary~\ref{cor:V^R matrix has full rank}, the matrix ${\mathbf{V}}^{\mathcal{R}_j}_{k,\ell,I,J}(\mathbf{X})$ has the full column rank. Therefore, the ``FAIL'' condition in Algorithm~\ref{alg:faulty index output} is never triggered.

    Next, suppose that Algorithm~\ref{alg:faulty index output} outputs ``SUCCESS''. Then, for some $j\in[r]$, the faulty index of the square matrix ${\mathbf{V}}^{\mathcal{R}_j}_{k,\ell,I,J,P_j}(\mathbf{X})$ does not exist. By Definition~\ref{def:faulty index}, the square submatrix ${\mathbf{V}}^{\mathcal{R}_j}_{k,\ell,I,J,P_j}(\mathbf{X})|_{\mathbf{X}=\boldsymbol{\alpha}}$ has full rank. Since this is a submatrix of $\mathbf{V}_{k,\ell,I,J}(\mathbf{X})|_{\mathbf{X}=\boldsymbol{\alpha}}$, the latter matrix has full column rank.

    Finally, suppose that the algorithm outputs a faulty index sequence $(u_1,\dots,u_r)\in[n]^r$. By the algorithm, for each $j\in[r]$, the index $u_j$ is the faulty index of ${\mathbf{V}}^{\mathcal{R}_j}_{k,\ell,I,J,P_j}(\mathbf{X})$, for $\mathcal{R}_j=\{u_1,\dots,u_{j-1}\}$. Indeed, for each $i\in \mathcal{R}_j$, all rows involving $X_i$ have already been deleted from ${\mathbf{V}}^{\mathcal{R}_j}_{k,\ell,I,J}(\mathbf{X})$, which implies that the index $i$ cannot be a faulty index of this submatrix. Therefore, $u_j\notin \mathcal{R}_{j}$ for each $j\in [r]$, which implies that the indices $u_1,\dots,u_r$ are distinct.
\end{proof}

The next lemma bounds the probability that Algorithm~\ref{alg:faulty index output} outputs a prescribed faulty index sequence over random $\boldsymbol{\alpha}$. The proof follows the same approach as~\cite[Lemma 19]{Con2024random}.

\begin{lemma}\label{lemma:probability V calculation fixed faulty}
    Under the notation and conditions in Lemma~\ref{lemma:behavior of alg V}, suppose that $\boldsymbol{\alpha}=(\alpha_1,\dots,\alpha_n)$ is chosen uniformly at random from the set of all $n$-tuples of distinct elements in $\mathbb{F}_q$. Then, for any fixed faulty index sequence $(u_1,\dots,u_r)\in[n]^r$, the probability that Algorithm~\ref{alg:faulty index output} outputs $(u_1,\dots,u_r)$ on input $\mathbf{V}_{k,\ell,I,J}(\mathbf{X})$, $\boldsymbol{\alpha}$, and $r$ is at most $\left(\frac{2(k-1)}{q-n+1}\right)^r$.
\end{lemma}

\begin{proof}
    For simplicity, denote $\mathbf{M}_j:={\mathbf{V}}^{\mathcal{R}_j}_{k,\ell,I,J,P_j}(\mathbf{X})$. Since the faulty index sequence $(u_1,\dots,u_r)$ is fixed, the matrices $\mathbf{M}_j$ are also fixed by the deterministic selection rule. Let $E_j$ be the event that
    $\det\left(\mathbf{M}_j|_{\mathbf{X}_{<u_j}=\boldsymbol{\alpha}_{<u_j}}\right)\neq 0$
    and
    $\det\left(\mathbf{M}_j|_{\mathbf{X}_{\le u_j}=\boldsymbol{\alpha}_{\le u_j}}\right)=0$.
    In other words, $E_j$ is the event that the faulty index of $\mathbf{M}_j$ equals $u_j$. Thus, if Algorithm~\ref{alg:faulty index output} outputs $(u_1,\dots,u_r)$, then all events $E_1,\dots,E_r$ occur. It therefore suffices to prove
    $$
    \textup{Pr}(E_1\wedge\cdots\wedge E_r)
    \le
    \left(\frac{2(k-1)}{q-n+1}\right)^r.
    $$

    Directly estimating conditional probabilities in the original order is inconvenient, because the faulty indices need not appear in increasing order. We therefore reorder them according to their values. Let $\sigma$ be a permutation of $[r]$ such that $u_{\sigma(1)}<\cdots<u_{\sigma(r)}$. For $h\in[r]$, define
    $F_h:=E_{\sigma(1)}\wedge\cdots\wedge E_{\sigma(h)}$,
    and let $F_0$ be the sure event. Then $\textup{Pr}(E_1\wedge\cdots\wedge E_r)=\textup{Pr}(F_r)$. If $\textup{Pr}(F_r)=0$, there is nothing to prove. Otherwise, $\textup{Pr}(F_h)>0$ for every $h\in[0,r]$, and hence
    $$
    \textup{Pr}(E_1\wedge\cdots\wedge E_r)
    =
    \textup{Pr}(F_r)
    =
    \prod_{h=1}^{r}
    \frac{\textup{Pr}(F_h)}{\textup{Pr}(F_{h-1})}.
    $$
    In what follows, we will show that $\textup{Pr}(F_h)/\textup{Pr}(F_{h-1})\le 2(k-1)/(q-n+1)$ for every $h\in[r]$. Once this is established, the proof is completed.

    Fix $h\in[r]$ and write $j=\sigma(h)$. Let $S$ be the set of all $\boldsymbol{\beta}=(\beta_1,\dots,\beta_{u_j-1})\in\mathbb{F}_q^{u_j-1}$ such that $\textup{Pr}((\mathbf{X}_{<u_j}=\boldsymbol{\beta})\wedge F_{h-1})>0$. By conditioning on the values of $\mathbf{X}_{<u_j}$ and using the definition of $S$, we have
    \begin{align*}
        \frac{\textup{Pr}(F_h)}{\textup{Pr}(F_{h-1})}
        &=
        \frac{\sum_{\boldsymbol{\beta}\in S}\textup{Pr}((\mathbf{X}_{<u_j}=\boldsymbol{\beta})\wedge F_{h-1}\wedge E_j)}
        {\sum_{\boldsymbol{\beta}\in S}\textup{Pr}((\mathbf{X}_{<u_j}=\boldsymbol{\beta})\wedge F_{h-1})} \\
        &\le
        \max_{\boldsymbol{\beta}\in S}
        \textup{Pr}(E_j\mid\mathbf{X}_{<u_j}=\boldsymbol{\beta}).
    \end{align*}
    It remains to show that for every fixed $\boldsymbol{\beta}\in S$,
    $\textup{Pr}(E_j\mid\mathbf{X}_{<u_j}=\boldsymbol{\beta})\le 2(k-1)/(q-n+1)$.

    Fix such a $\boldsymbol{\beta}$ and set $Q:=\det\left(\mathbf{M}_j|_{\mathbf{X}_{<u_j}=\boldsymbol{\beta}}\right)$. If $Q=0$, then $E_j$ cannot occur under the conditioning $\mathbf{X}_{<u_j}=\boldsymbol{\beta}$. Suppose now that $Q\neq 0$. We view $Q$ as a polynomial in $X_{u_j+1},\dots,X_n$, whose coefficients are polynomials in $\mathbb{F}_q[X_{u_j}]$. Since $Q\neq 0$, there exists a nonzero coefficient $Q_0\in\mathbb{F}_q[X_{u_j}]$. Moreover, since $\mathbf{M}_j$ is a square submatrix of a $\mathbf{V}$-matrix, we have $\deg(Q_0)\le 2(k-1)$.

    Conditioned on $\mathbf{X}_{<u_j}=\boldsymbol{\beta}$, the value $\alpha_{u_j}$ is chosen uniformly from at least $q-n+1$ available field elements. Event $E_j$ forces $Q|_{X_{u_j}=\alpha_{u_j}}$ to vanish as a polynomial in the remaining variables, and hence $Q_0(\alpha_{u_j})=0$. By Lemma~\ref{lemma:Schwarz-Zippel variation},
    $$
    \textup{Pr}(E_j\mid\mathbf{X}_{<u_j}=\boldsymbol{\beta})
    \le
    \textup{Pr}(Q_0(\alpha_{u_j})=0)
    \le
    \frac{2(k-1)}{q-n+1}.
    $$
    This completes the proof.
\end{proof}

Combining Lemma~\ref{lemma:behavior of alg V} and Lemma~\ref{lemma:probability V calculation fixed faulty}, we obtain the following corollary.

\begin{cor}\label{cor:probability V arbitrary faulty}
    Under the notation and conditions in Lemma~\ref{lemma:behavior of alg V}, suppose that $q\ge n$ and that $\boldsymbol{\alpha}$ is chosen uniformly at random from the set of all $n$-tuples of distinct elements in $\mathbb{F}_q$. Let $r=(\epsilon-\rho)n/2$. Then
    $$
        \textup{Pr}\left(\mathbf{V}_{k,\ell,I,J}(\mathbf{X})|_{\mathbf{X}=\boldsymbol{\alpha}}\textup{ does not have full column rank}\right)
        \le
        \left(\frac{2n(k-1)}{q-n+1}\right)^r.
    $$
\end{cor}

\begin{proof}
    By Lemma~\ref{lemma:behavior of alg V}, if $\mathbf{V}_{k,\ell,I,J}(\mathbf{X})|_{\mathbf{X}=\boldsymbol{\alpha}}$ does not have full column rank, then Algorithm~\ref{alg:faulty index output} must output a faulty index sequence $(u_1,\dots,u_r)\in[n]^r$. For each fixed sequence, Lemma~\ref{lemma:probability V calculation fixed faulty} bounds this probability by $\left(\frac{2(k-1)}{q-n+1}\right)^r$. Taking a union bound over all $n^r$ possible sequences gives the desired bound.
\end{proof}

It follows from Corollary~\ref{cor:probability V arbitrary faulty} that the $\mathbf{V}$-matrix rank condition holds simultaneously for all relevant pairs $(I,J)$ with high probability.

\begin{prop}\label{prop:V matrix probability for all I,J}
    Let $\rho,\epsilon\in(0,1)$ with $0<\rho<\epsilon$, and $k\le n$ be positive integers with $k=Rn$. Let $\ell=2k-1+\epsilon n$, and suppose that $2R+\epsilon<1$. Let $q$ be a prime power such that
    $$
    q\ge n+2\cdot 2^{\frac{6}{\epsilon-\rho}}
    ((2R+\epsilon)n)^{\frac{4R+2\epsilon}{\epsilon-\rho}}
    \cdot nk.
    $$
    Suppose that $\boldsymbol{\alpha}=(\alpha_1,\dots,\alpha_n)$ is chosen uniformly at random from the set of all $n$-tuples of distinct elements in $\mathbb{F}_q$. Then, with probability at least $1-2^{-n}$, for every pair of distinct-element sequences $I,J\in[n]^\ell$ containing at most $k+\rho n$ minimal equal sub-pairs, the matrix
    $\mathbf{V}_{k,\ell,I,J}(\mathbf{X})|_{\mathbf{X}=\boldsymbol{\alpha}}$
    has full column rank.
\end{prop}

\begin{proof}
    Let $I=(i_1,\dots,i_{\ell})$ and $J=(j_1,\dots,j_\ell)$, and let $\pi$ be a permutation of $[\ell]$. Define
    $\pi(I)=(i_{\pi(1)},\dots,i_{\pi(\ell)})$
    and
    $\pi(J)=(j_{\pi(1)},\dots,j_{\pi(\ell)})$.
    The matrices $\mathbf{V}_{k,\ell,I,J}(\mathbf{X})$ and $\mathbf{V}_{k,\ell,\pi(I),\pi(J)}(\mathbf{X})$ differ only by a permutation of rows. Hence
    $$
    \operatorname{rank}\left(\mathbf{V}_{k,\ell,I,J}(\mathbf{X})|_{\mathbf{X}=\boldsymbol{\alpha}}\right)
    =
    \operatorname{rank}\left(\mathbf{V}_{k,\ell,\pi(I),\pi(J)}(\mathbf{X})|_{\mathbf{X}=\boldsymbol{\alpha}}\right).
    $$

    We can identify two pairs $(I,J)$ and $(I',J')$ if there exists a permutation $\pi$ of $[\ell]$ such that $I'=\pi(I)$ and $J'=\pi(J)$. The number of pairs of distinct-element sequences of length $\ell$ is at most $\binom{n}{\ell}^2(\ell!)^2$. After quotienting by the above equivalence relation, the number of equivalence classes is at most $\binom{n}{\ell}^2\ell!$.

    For each fixed equivalence class, Corollary~\ref{cor:probability V arbitrary faulty}, with $r=(\epsilon-\rho)n/2$, gives
    $$
    \textup{Pr}\left(
    \mathbf{V}_{k,\ell,I,J}(\mathbf{X})|_{\mathbf{X}=\boldsymbol{\alpha}}
    \textup{ does not have full column rank}
    \right)
    \le
    \left(\frac{2n(k-1)}{q-n+1}\right)^{\frac{(\epsilon-\rho)n}{2}}.
    $$
    Thus, by the union bound, the probability of failure over all relevant pairs $(I,J)$ is at most
    $$
    \binom{n}{\ell}^2\ell!
    \left(\frac{2n(k-1)}{q-n+1}\right)^{\frac{(\epsilon-\rho)n}{2}}
    \le
    2^{2n}((2R+\epsilon)n)^{(2R+\epsilon)n}
    \left(\frac{2n(k-1)}{q-n+1}\right)^{\frac{(\epsilon-\rho)n}{2}}.
    $$
    Here, the inequality follows from $\binom{n}{\ell}\le 2^n$, $\ell!\le \ell^\ell$, and $\ell\le (2R+\epsilon)n$. By the assumed lower bound on $q$, the last quantity is at most $2^{-n}$. This completes the proof.
\end{proof}

\begin{rmk}
    The factorial term $\ell!$ in the union bound reflects a basic difficulty of the permutation-insdel setting: unlike ordinary insdel errors, the selected index sequences are not required to be increasing. Thus, within the present faulty index and union bound framework, this term appears to be the main obstruction to obtaining a linear-size alphabet as in~\cite{Con2024random}. Removing this factorial loss would likely require a different idea.
\end{rmk}

\subsection{$\mathbf{A}$-matrix}\label{section:A matrix}

In this section, we analyze the $\mathbf{A}$-matrix regime. Compared with the $\mathbf{V}$-matrix regime, the main advantage is that $\mathbf{A}$-matrices satisfy a robust full-rank property, which allows us to select square submatrices in a direct deterministic way.

By Lemma~\ref{lemma:Amatrix full rank}, when applying Algorithm~\ref{alg:faulty index output} to $\mathbf{A}_{k,I^1,\dots,I^t}(\mathbf{X})$, we may select the square submatrix in Line 6 in a deterministic way. Specifically, whenever $\mathbf{A}_{k,I^1,\dots,I^t}^\mathcal{R}(\mathbf{X})$ has at least $k$ rows, let $\mathbf{A}_P^\mathcal{R}(\mathbf{X})$ denote the $k\times k$ submatrix consisting of the first $k$ rows of $\mathbf{A}_{k,I^1,\dots,I^t}^\mathcal{R}(\mathbf{X})$, with respect to the fixed order $I^1,\dots,I^t$. We now combine this deterministic selector with Algorithm~\ref{alg:faulty index output}.

\begin{lemma}\label{lemma:behavior of alg A}
    Let $0<\rho<\epsilon<1$, $\ell=2k-1+\epsilon n$, and $r'=\rho n$. Let $I^1,\dots,I^t\subseteq[n]$ be pairwise disjoint nonempty sets with $t\ge k+\rho n$. Suppose that $|I^1|+\dots +|I^t|\le \ell$ and $\operatorname{char}(\mathbb{F}_q)>\ell$. Then, for any pairwise distinct $\alpha_1,\dots,\alpha_n\in\mathbb{F}_q$, running Algorithm~\ref{alg:faulty index output} on the matrix $\mathbf{A}_{k,I^1,\dots,I^t}(\mathbf{X})$, while always selecting the $k\times k$ submatrix $\mathbf{A}_P^\mathcal{R}(\mathbf{X})$ consisting of the first $k$ surviving rows in Line 6, yields one of the following two outcomes:
    \begin{itemize}
        \item[1.] Algorithm~\ref{alg:faulty index output} outputs ``SUCCESS''. In this case, $\mathbf{A}_{k,I^1,\dots,I^t}(\mathbf{X})|_{\mathbf{X}=\boldsymbol{\alpha}}$ has the full column rank.
        \item[2.] Algorithm~\ref{alg:faulty index output} outputs a faulty index sequence $(u_1,\dots,u_{r'})\in[n]^{r'}$. The indices $u_1,\dots,u_{r'}$ are distinct. Moreover, for each $j\in[r']$, $u_j$ is the faulty index of $\mathbf{A}^{\mathcal{R}_j}_{P}(\mathbf{X})$, where $\mathcal{R}_j:=\{u_1,\dots,u_{j-1}\}$ and $\mathbf{A}^{\mathcal{R}_j}_{P}(\mathbf{X})$ is the first $k\times k$ submatrix of $\mathbf{A}^{\mathcal{R}_j}_{k,I^1,\dots,I^t}(\mathbf{X})$.
    \end{itemize}
\end{lemma}
\begin{proof}
    The proof of Lemma~\ref{lemma:behavior of alg A} is analogous to that of Lemma~\ref{lemma:behavior of alg V}. Indeed, since $t\ge k+\rho n$, $r'=\rho n$ and the disjoint-ness of $I^1,\dots,I^t$, after deleting rows involving at most $r'$ faulty indices, the remaining $\mathbf{A}$-matrix still has at least $k$ rows. Therefore, by Lemma~\ref{lemma:Amatrix full rank}, the ``FAIL'' condition in Algorithm~\ref{alg:faulty index output} is never triggered. The remaining two conclusions follow exactly as in the proof of Lemma~\ref{lemma:behavior of alg V}, and we omit the details.
\end{proof}

The next lemma is the $\mathbf{A}$-matrix analog of Lemma~\ref{lemma:probability V calculation fixed faulty}. It bounds the probability that Algorithm~\ref{alg:faulty index output} outputs a prescribed faulty index sequence over random $\boldsymbol{\alpha}$.

\begin{lemma}\label{lemma:probability A calculation fixed faulty}
    Under the notation and conditions in Lemma~\ref{lemma:behavior of alg A}, suppose that $\boldsymbol{\alpha}=(\alpha_1,\dots,\alpha_n)$ is chosen uniformly at random from the set of all $n$-tuples of distinct elements in $\mathbb{F}_q$. Then, for any fixed faulty index sequence $(u_1,\dots,u_{r'})\in[n]^{r'}$, the probability that Algorithm~\ref{alg:faulty index output} outputs $(u_1,\dots,u_{r'})$ on input $\mathbf{A}_{k,I^1,\dots,I^t}(\mathbf{X})$, $\boldsymbol{\alpha}$, and $r'$ is at most $\left(\frac{k-1}{q-n+1}\right)^{r'}$.
\end{lemma}

\begin{proof}
    The proof is analogous to that of Lemma~\ref{lemma:probability V calculation fixed faulty}. The only difference is that, by Lemma~\ref{lemma:Amatrix full rank}, the degree of each variable in the relevant determinant is at most $k-1$. Hence the same faulty index argument gives a factor of $(k-1)/(q-n+1)$ at each step, and the desired bound follows.
\end{proof}

The next lemma is specific to the $\mathbf{A}$-matrix regime. It improves the crude bound $n^{r'}$ on the number of possible faulty index sequences to $\binom{n}{r'}$ by exploiting the pairwise disjoint supports of the rows of an $\mathbf{A}$-matrix.

\begin{lemma}\label{lemma:number of possible index A}
    Under the notation and conditions in Lemma~\ref{lemma:behavior of alg A}, let $\mathcal{T}$ be the set of all possible faulty index sequences $(u_1,\dots,u_{r'})$ output by Algorithm~\ref{alg:faulty index output}. Then $|\mathcal{T}|\le \binom{n}{r'}$.
\end{lemma}

\begin{proof}
    Since the sets $I^1,\dots,I^t$ are pairwise disjoint, a relabeling of the variables induces a bijection on the set of possible faulty index sequences. Thus we may assume that every element of $I^a$ is smaller than every element of $I^b$ whenever $a<b$.

    We show that, after this relabeling, every possible faulty index sequence is strictly increasing. Fix one step of Algorithm~\ref{alg:faulty index output}, and let $u$ be the faulty index found at this step. Let $\mathbf{M}$ be the selected $k\times k$ submatrix before $u$ is recorded. If $u$ does not belong to one of the selected rows, then $\det(\mathbf{M})$ is independent of $X_u$, and $u$ cannot be faulty. Hence $u$ belongs to some selected row, say the row corresponding to a set $S\in \{I^1,\dots,I^t\}$.

    By Definition~\ref{def:faulty index}, $\det\left(\mathbf{M}\big|_{\mathbf{X}_{<u}=\boldsymbol{\alpha}_{<u}}\right)\neq 0$. Expanding this determinant along the row corresponding to $S$, at least one signed cofactor is a nonzero polynomial. These cofactors are formed from the other $k-1$ selected rows. Since sets $I^1,\dots,I^t$ are pairwise disjoint, they do not involve any variable from $S$, and in particular they are unchanged by assignment $X_u=\alpha_u$.

    After $u$ is recorded, the row corresponding to $S$ is deleted. If the algorithm continues, the first-$k$-surviving-row rule adds a later row, say the row corresponding to a set $S'$. By our relabeling convention, all variables in $S'$ have indices larger than $u$. Write $S'=\{v_1,\dots,v_m\}$ and set $Y_a=X_{v_a}$. The new row is
    $$
    \left(m,\sum_{a=1}^{m}Y_a,\sum_{a=1}^{m}Y_a^2,\dots,\sum_{a=1}^{m}Y_a^{k-1}\right).
    $$

    Now consider the determinant of the next selected submatrix after the faulty index $u$ is recorded. Expanding it along the new row gives
    $$
    mC_0+\sum_{d=1}^{k-1}\left(\sum_{a=1}^{m}Y_a^d\right)C_d,
    $$
    where the $C_d$'s are signed cofactors formed from the previous $k-1$ selected rows. These cofactors agree, up to signs, with the cofactors obtained in the previous expansion along the row $S$. Hence not all $C_d$ are zero. Let $d_0$ be the largest index such that $C_{d_0}\neq 0$. If $d_0\ge 1$, then the coefficient of $Y_1^{d_0}$ is $C_{d_0}\neq 0$. If $d_0=0$, then the determinant is $mC_0\neq 0$, since $1\le m\le \ell$ and $\operatorname{char}(\mathbb{F}_q)>\ell$.

    Thus the next selected determinant is not the zero polynomial after the partial assignment $\mathbf{X}_{\le u}=\boldsymbol{\alpha}_{\le u}$. Therefore the next faulty index, if it exists, must be larger than $u$. Hence every possible faulty index sequence is strictly increasing. Consequently, $|\mathcal{T}|\le \binom{n}{r'}$.
\end{proof}

Combining Lemma~\ref{lemma:behavior of alg A}, Lemma~\ref{lemma:probability A calculation fixed faulty}, and Lemma~\ref{lemma:number of possible index A}, we obtain the following corollary.

\begin{cor}\label{cor:probability A arbitrary faulty}
    Under the notation and conditions in Lemma~\ref{lemma:number of possible index A}, suppose that $q\ge n$ and that $\boldsymbol{\alpha}$ is chosen uniformly at random from the set of all $n$-tuples of distinct elements in $\mathbb{F}_q$. Let $r'=\rho n$. Then
    $$
        \textup{Pr}\left(\mathbf{A}_{k,I^1,\dots,I^t}(\mathbf{X})|_{\mathbf{X}=\boldsymbol{\alpha}}\textup{ does not have full column rank}\right)
        \le
        2^n\left(\frac{k-1}{q-n+1}\right)^{r'}.
    $$
\end{cor}

\begin{proof}
    By Lemma~\ref{lemma:behavior of alg A}, if $\mathbf{A}_{k,I^1,\dots,I^t}(\mathbf{X})|_{\mathbf{X}=\boldsymbol{\alpha}}$ does not have full column rank, then Algorithm~\ref{alg:faulty index output} must output a faulty index sequence. For each fixed faulty index sequence, Lemma~\ref{lemma:probability A calculation fixed faulty} gives an upper bound $\left(\frac{k-1}{q-n+1}\right)^{r'}$. By Lemma~\ref{lemma:number of possible index A}, there are at most $\binom{n}{r'}$ possible faulty index sequences. The desired bound follows from the union bound and $\binom{n}{r'}\le 2^n$.
\end{proof}

It follows from Corollary~\ref{cor:probability A arbitrary faulty} that the $\mathbf{A}$-matrix rank condition holds simultaneously for all relevant collections $I^1,\dots,I^t$ with high probability.

\begin{prop}\label{prop:A matrix probability for all I,J}
    Let $\rho,\epsilon\in(0,1)$ with $0<\rho<\epsilon$, and $k\le n$ be positive integers with $k=Rn$. Let $\ell=2k-1+\epsilon n$, and suppose that $2R+\epsilon<1$. Let $q$ be a prime power such that
    $$
    q\ge n+2^{\frac{3}{\rho}}(k+\rho n)^{\frac{2R+\epsilon}{\rho}}k,
    $$
    and $\operatorname{char}(\mathbb F_q)>\ell$. Suppose that $\boldsymbol{\alpha}=(\alpha_1,\dots,\alpha_n)$ is chosen uniformly at random from the set of all $n$-tuples of distinct elements in $\mathbb{F}_q$. Then, with probability at least $1-2^{-n}$, for $t=k+\rho n$ and every collection of pairwise disjoint nonempty sets $I^1,\dots,I^t\subseteq[n]$ satisfying $|I^1|+\cdots+|I^t|\le \ell$, the $t\times k$ matrix
    $\mathbf{A}_{k,I^1,\dots,I^t}(\mathbf{X})|_{\mathbf{X}=\boldsymbol{\alpha}}$
    has full column rank.
\end{prop}

\begin{proof}
    We first count the number of possible ordered collections $I^1,\dots,I^t$. Let
    $m=|I^1|+\cdots+|I^t|$ and write $s_j=|I^j|$. Then $t\le m\le \ell$. For fixed $m$ and fixed positive integers $s_1,\dots,s_t$ with $s_1+\cdots+s_t=m$, the number of collections with $|I^j|=s_j$ for every $j\in[t]$ is at most
    $$
    \binom{n}{m}\binom{m}{s_1,\dots,s_t}.
    $$
    Therefore, the total number of relevant collections is at most
    \begin{equation*}
        \sum_{m=t}^{\ell}\binom{n}{m}
        \sum_{\substack{s_1+\cdots+s_t=m\\ s_j\ge 1,\ \forall j\in[t]}}
        \binom{m}{s_1,\dots,s_t} \le
        \sum_{m=t}^{\ell}\binom{n}{m}t^m  \le
        2^n t^{\ell+1},
    \end{equation*}
    where the first inequality follows from the multinomial theorem.

    For each fixed collection, Corollary~\ref{cor:probability A arbitrary faulty}, with $r'=\rho n$, gives
    $$
    \textup{Pr}\left(
    \mathbf{A}_{k,I^1,\dots,I^t}(\mathbf{X})|_{\mathbf{X}=\boldsymbol{\alpha}}
    \textup{ does not have full column rank}
    \right)
    \le
    2^n\left(\frac{k-1}{q-n+1}\right)^{\rho n}.
    $$
    Therefore, by the union bound, the failure probability over all relevant collections is at most $2^n t^{\ell+1}\cdot 2^n\left(\frac{k-1}{q-n+1}\right)^{\rho n}$. Since $t=k+\rho n$ and $\ell+1=(2R+\epsilon)n$, the assumed lower bound on $q$ implies that this quantity is at most $2^{-n}$. This proves the proposition.
\end{proof}

\begin{rmk}
    In the $\mathbf{A}$-matrix regime, the dominant term in the union bound is $t^{\ell+1}$, which comes from counting the possible collections of disjoint sets $I^1,\dots,I^t$. This counting loss is the main obstruction in the present approach; improving the alphabet size would require a different method.
\end{rmk}

\subsection{Proof of Theorem~\ref{thm:main informal}}\label{sec:put it together}

We now combine the estimates from the $\mathbf{V}$-matrix and $\mathbf{A}$-matrix regimes to prove the main theorem.

\begin{thm}\label{thm:main formal}
    Let $k$ and $n$ be positive integers with $k=Rn$, where $R,\epsilon\in(0,1)$ are constants satisfying $2R+\epsilon<1$. Set
    $$
    \rho=3R+2\epsilon-\sqrt{9R^2+10R\epsilon+3\epsilon^2}.
    $$
    Then $0<\rho<\epsilon$. Let $q$ be a prime power such that $\operatorname{char}(\mathbb F_q)>2k-1+\epsilon n$ and $q\ge n+C_1 n^{C_2}$, where $C_1$ and $C_2$ are positive constants depending only on $R$ and $\epsilon$. More explicitly,
    $$
    C_1=
    \max\left\{
    2R\cdot 2^{\frac{6}{\epsilon-\rho}}(2R+\epsilon)^{\frac{4R+2\epsilon}{\epsilon-\rho}},
    R\cdot 2^{\frac{3}{\rho}}(R+\rho)^{\frac{2R+\epsilon}{\rho}}
    \right\},
    $$
    and
    $$
    C_2=
    \frac{3R+3\epsilon+\sqrt{9R^2+10R\epsilon+3\epsilon^2}}{\epsilon}.
    $$
    Then a random $\textup{RS}_{n,k}(\boldsymbol{\alpha})$ code over $\mathbb{F}_q$ is, with probability at least $1-2^{-n+1}$, robust against the $((1-\epsilon)n-2k+1)$-permutation-insdel adversary.
\end{thm}

\begin{proof}
    We first check that $0<\rho<\epsilon$. Indeed,
    $$
    (3R+\epsilon)^2
    <
    9R^2+10R\epsilon+3\epsilon^2
    <
    (3R+2\epsilon)^2.
    $$
    Hence
    $$
    3R+\epsilon
    <
    \sqrt{9R^2+10R\epsilon+3\epsilon^2}
    <
    3R+2\epsilon,
    $$
    which gives $0<\rho<\epsilon$.

    By Lemma~\ref{lemma:V A two conditions}, it remains to verify the two required rank conditions. The choice of $\rho$ makes the alphabet-size exponents in Propositions~\ref{prop:V matrix probability for all I,J} and~\ref{prop:A matrix probability for all I,J} both equal to $C_2$, while $C_1$ dominates the corresponding leading constants. Hence the assumption $q\ge n+C_1n^{C_2}$ allows us to apply both propositions. A union bound gives simultaneous validity of the two rank conditions with probability at least $1-2^{-n+1}$, and the rank criterion then yields the claimed robustness of $((1-\epsilon)n-2k+1)$-permutation-insdel.

\end{proof}

\subsection{Alphabet lower bound}\label{sec:alphabet lower bound}

In this section, we prove an alphabet-size lower bound for codes robust against permutation-insdel adversaries. The argument is purely combinatorial and does not use any algebraic structure of the code. In particular, it applies first to arbitrary $q$-ary codes, and the Reed--Solomon lower bound follows as a special case.

Let $\Sigma$ be an alphabet of size $q$. We consider a general code $\mathcal{C}\subseteq\Sigma^n$, not necessarily linear. In a permutation channel, the order of the coordinates is irrelevant. Thus it is natural to represent a received deletion output by its symbol multiplicities. This leads to the following histogram formulation.

For a word $\mathbf{c}=(c_1,\dots,c_n)\in\Sigma^n$, define its \emph{histogram} by $\mathrm{Hist}(\mathbf{c})=(m_a(\mathbf{c}))_{a\in\Sigma}\in \mathbb{Z}_{\ge 0}^q$, where $m_a(\mathbf{c})=\left|\{i\in[n]:c_i=a\}\right|$. Thus $\mathrm{Hist}(\mathbf{c})$ lies in the \emph{discrete simplex}
$$
\Delta_{n,q}:=
\left\{\mathbf{m}=(m_a)_{a\in\Sigma}\in\mathbb{Z}_{\ge 0}^{q}:
\sum_{a\in\Sigma}m_a=n\right\}.
$$
More generally, the number of histograms of total weight $m$ over a $q$-symbol alphabet is $|\Delta_{m,q}|=\binom{q+m-1}{m}$.

A deletion decreases one coordinate of the histogram by one. Hence, after exactly $t$ deletions, the received histogram has total weight $n-t$.

For $\mathbf{m}\in\Delta_{n,q}$ and $0\le t\le n$, define the \emph{exact $t$-deletion ball} of $\mathbf{m}$ by
$$
\mathcal{D}_t(\mathbf{m})
:=
\{\mathbf{y}\in\Delta_{n-t,q}:0\le y_a\le m_a
\text{ for every } a\in\Sigma\}.
$$
Equivalently, $\mathbf{y}\in\mathcal{D}_t(\mathbf{m})$ if and only if $\mathbf{y}$ is obtained from $\mathbf{m}$ by decreasing its coordinates by a nonnegative vector of total weight $t$.

\begin{lemma}\label{lemma:alphabet general}
    Let $\mathcal{C}\subseteq\Sigma^n$ be a code that is robust against the $t$-permutation-insdel adversary. Then $|\mathcal{C}|\le \binom{q+n-t-1}{n-t}$.
\end{lemma}

\begin{proof}
    Set $\ell=n-t$. Since the adversary is allowed to perform $t$ deletions, no two distinct codewords can produce the same length-$\ell$ received histogram after deleting $t$ symbols. Equivalently, for any distinct codewords $\mathbf{c},\mathbf{c}'\in\mathcal{C}$, $$ \mathcal{D}_t(\mathrm{Hist}(\mathbf{c})) \cap \mathcal{D}_t(\mathrm{Hist}(\mathbf{c}')) = \varnothing. $$ Indeed, if a histogram $\mathbf{y}$ belonged to both deletion balls, then after suitable coordinate permutations and $t$ deletions, both codewords could yield the same received word, whose histogram equals $\mathbf{y}$, contradicting robustness.

    All exact $t$-deletion histograms lie in the same discrete simplex $\Delta_{\ell,q}$. Since the deletion balls are pairwise disjoint subsets of $\Delta_{\ell,q}$, we have
    $$
    \sum_{\mathbf{c}\in\mathcal{C}}
    |\mathcal{D}_t(\mathrm{Hist}(\mathbf{c}))|
    \le
    |\Delta_{\ell,q}|.
    $$
    Each exact deletion ball is nonempty, so
    $$
    |\mathcal{C}|
    \le
    \sum_{\mathbf{c}\in\mathcal{C}}
    |\mathcal{D}_t(\mathrm{Hist}(\mathbf{c}))|
    \le
    |\Delta_{\ell,q}|
    =
    \binom{q+\ell-1}{\ell}.
    $$
    Since $\ell=n-t$, the desired bound follows.
\end{proof}

We now prove the alphabet-size lower bound stated in Theorem~\ref{thm:alphabet lower informal}, which we restate for convenience. The condition $R+\delta< 1$ is relative to the Singleton-type bound $k+t\le n$ when $k=Rn$ and $t=\delta n$. Here, for convenience, we assume $Rn$ and $\delta n$ to be integers.

\begin{thm}\label{thm:alphabet lowerbound}
    Let $R,\delta\in(0,1)$ be constants satisfying $R+\delta< 1$. Let $\mathcal{C}\subseteq \Sigma^n$ be a $q$-ary code with $|\mathcal{C}|=q^k$ and $k=Rn$. Suppose that $\mathcal{C}$ is robust against the $t$-permutation-insdel adversary, where $t=\delta n$. Then, for all sufficiently large $n$, the alphabet size satisfies
    $$
    q\ge C_{R,\delta}\, n^{\frac{1-\delta}{1-\delta-R}},
    $$
    where $C_{R,\delta}>0$ depends only on $R$ and $\delta$.
\end{thm}

\begin{proof}
    Set $\ell=n-t=(1-\delta)n$. By Lemma~\ref{lemma:alphabet general},
    $$
    q^k=|\mathcal{C}|\le \binom{q+\ell-1}{\ell}.
    $$

    We first show that $q$ cannot be bounded by a constant. If $q=O(1)$, then $\binom{q+\ell-1}{\ell}$ is polynomial in $n$, whereas $q^k=q^{Rn}$ is exponential in $n$, a contradiction for all sufficiently large $n$.

    We next show that $q=\Omega(n)$. Suppose, to the contrary, that $q=o(n)$. By the previous paragraph, we may assume that $q\to\infty$. Since $\ell=\Theta(n)$, 
    $$
    \binom{q+\ell-1}{\ell}\le 2^{q+\ell}=2^{O(n)}.
    $$
    It follows that $k\log q=O(n)$, so $k=O(n/\log q)=o(n)$, contradicting $k=Rn$.

    Hence $q=\Omega(n)$, and there exists a constant $C_0=C_0(R,\delta)>0$ such that $n\le C_0q$ for all sufficiently large $n$. Consequently,
    $$
    q^k\le \binom{q+\ell-1}{\ell}
    \le
    \left(\frac{e(q+\ell)}{\ell}\right)^\ell
    \le
    \left(\frac{e(q+n)}{\ell}\right)^\ell
    \le
    \left(\frac{e(1+C_0)q}{\ell}\right)^\ell.
    $$
    Rearranging gives
    $$
    q\ge
    \left(\frac{\ell}{e(1+C_0)}\right)^{\frac{\ell}{\ell-k}}.
    $$
    Since $\ell=(1-\delta)n$ and $k=Rn$, for all sufficiently large $n$, we have
    $$
    q\ge C_{R,\delta}\, n^{\frac{1-\delta}{1-\delta-R}},
    $$
    for some constant $C_{R,\delta}>0$ depending only on $R$ and $\delta$.
\end{proof}

\begin{rmk}
    In Lemma~\ref{lemma:alphabet general}, we used only the trivial bound
    $|\mathcal{D}_t(\mathrm{Hist}(\mathbf{c}))|\ge 1$. This already suffices to obtain the exponent in Theorem~\ref{thm:alphabet lowerbound}. Even if one could obtain a uniform exponential lower bound on the sizes of the deletion balls, such an improvement would affect only the constant $C_{R,\delta}$, not the exponent of $n$.
\end{rmk}

As a direct consequence, we obtain the following lower bound for Reed--Solomon codes.

\begin{cor}\label{cor:alphabet lower RS}
    Fix constants $R\in(0,1)$ and $\epsilon\in[0,1)$ such that $2R+\epsilon<1$, and let $k=Rn$. Suppose that the Reed--Solomon code $\mathrm{RS}_{n,k}(\boldsymbol{\alpha})$ over $\mathbb{F}_q$ is robust against the $((1-\epsilon)n-2k+1)$-permutation-insdel adversary. Then, for all sufficiently large $n$,
    $$
    q\ge C_{R,\epsilon}\,n^{\frac{2R+\epsilon}{R+\epsilon}},
    $$
    where $C_{R,\epsilon}>0$ depends only on $R$ and $\epsilon$.
\end{cor}

Thus, unlike the ordinary insdel setting, one cannot hope for a linear-size alphabet in the permutation-insdel setting for constant-rate Reed--Solomon codes.

\section{Optimal Decoding Algorithm}\label{sec:decoding}

In this section, we give an average $O(n)$-time decoder in the permutation-insdel setting for the explicit two-dimensional Reed--Solomon codes constructed in~\cite{Con2024twodimension}. More precisely, for these codes, the decoder recovers the transmitted codeword from any received word obtained by an arbitrary coordinate permutation followed by at most $n-3$ insdel errors.

\subsection{Singhvi's deletion decoder}\label{section:Singhvi decoding}

We first recall the explicit two-dimensional Reed--Solomon codes of~\cite[Proposition 2.3]{Con2024twodimension} and the linear-time deletion decoder of~\cite{Singhvi2026twodimension}. The construction gives $\textup{RS}_{n,2}$ codes over $\mathbb{F}_{q^3}$ with block length up to $q-1$, and hence alphabet size $O(n^3)$; these codes can correct $n-3$ insdel errors. Later in~\cite[Theorem 4]{Con2025anonymous} it was shown that the same codes are robust against the $(n-3)$-permutation-insdel adversary. However, the decoder of~\cite{Singhvi2026twodimension} applies only to the deletion setting: given any three surviving genuine symbols, it recovers the transmitted codeword in $O(n)$ time.

\begin{constr}\cite{Con2024twodimension,Con2025anonymous}\label{constr:2dim} 
    Let $q$ be an odd prime power, and $\Lambda\subseteq \mathbb{F}_q\setminus\{0\}$ be a subset of size $n$. Let $\gamma$ be a root
    of a degree-$3$ irreducible polynomial over $\mathbb{F}_q$. Let
    $\boldsymbol{\alpha}=(\alpha_1,\alpha_2,\dots,\alpha_n)$ be an ordering of
    the $n$ elements $a+a^2\gamma$, where $a\in \Lambda$. Then
    $\textup{RS}_{n,2}(\boldsymbol{\alpha})$ over $\mathbb{F}_{q^3}$ is robust against the $(n-3)$-permutation-insdel adversary. Moreover, the block length $n$ can be as large as $q-1$. 
\end{constr}

\begin{lemma}\cite[Theorem 2]{Singhvi2026twodimension}\label{lemma:2dim alg deletion}
    Let $\mathcal{C}=\textup{RS}_{n,2}(\boldsymbol{\alpha})$ be the code over
    $\mathbb{F}_{q^3}$ specified in Construction~\ref{constr:2dim}. Then there exists a decoder
    that recovers the transmitted codeword from any $n-3$ deletions in $O(n)$ time.
\end{lemma}

We briefly recall  the deletion decoder that will be used later. Suppose exactly three genuine symbols
$(c_{\kappa_1},c_{\kappa_2},c_{\kappa_3})$ remain, where $\kappa_1<\kappa_2<\kappa_3$. For a nonconstant affine message polynomial $f(x)=m_0+m_1x$ with $m_1\neq 0$, we have
$$
  \beta
  :=
  \frac{c_{\kappa_1}-c_{\kappa_2}}{c_{\kappa_2}-c_{\kappa_3}}
  =
  \frac{\alpha_{\kappa_1}-\alpha_{\kappa_2}}
       {\alpha_{\kappa_2}-\alpha_{\kappa_3}}.
$$
For this construction, the map
$$
  \Gamma(\alpha_i,\alpha_j,\alpha_k)
  :=
  \frac{\alpha_i-\alpha_j}{\alpha_j-\alpha_k}
$$
is injective on triples $i<j<k$. Using the special form $\alpha_i=a_i+a_i^2\gamma$, the decoder can invert this map in $O(1)$ field operations, with a hash table storing $a_i\mapsto i$. Thus the inverse-ratio subroutine recovers $(\kappa_1,\kappa_2,\kappa_3)$ from $\beta$ in $O(1)$ time. Once these three indices are known, interpolating the affine polynomial also takes $O(1)$ time, while reconstructing or outputting the full codeword takes $O(n)$ time.

\subsection{Decoder for permutations and insdel errors}

In this section, we give a decoding algorithm for the two-dimensional Reed--Solomon code in Construction~\ref{constr:2dim} against an arbitrary coordinate permutation followed by at most $n-3$ insdel errors. A straightforward approach would be to apply the deletion decoder to every triple of symbols in the received word. Each triple produces at most one candidate affine polynomial $\hat f(x)$, and the transmitted polynomial $f(x)$ could then be identified by a majority vote among these candidates. This brute-force strategy, however, requires an examination of $O(n^3)$ triples and therefore has cubic time complexity. We show that this exhaustive search can be avoided, leading to an average $O(n)$-time decoder.

Let $\mathbf{c}_f$ be the transmitted codeword, corresponding to the affine function $f$, and let $\mathbf{y}=(y_1,\ldots,y_m)$ be the received word, with $m\in [3,2n-3]$. 

\begin{defn}[Genuine position and genuine triple]
    Let $\mathbf{y}=(y_1,\ldots,y_m)$ be a received word obtained from a transmitted codeword $\mathbf{c}_f$ by permutations and insdel errors. A position of $\mathbf{y}$ is called a \emph{genuine position} if the corresponding symbol is not inserted, i.e., it comes from a non-deleted symbol of the permuted transmitted codeword.

    Given three distinct genuine positions $i,j,k\in[m]$, we call $(y_i,y_j,y_k)$ a \emph{genuine triple} if it appears as a subsequence of $\mathbf{c}_f$. Equivalently, these three positions not only originate from the transmitted codeword, but also preserve their relative order in $\mathbf{c}_f$.
\end{defn}

In the next lemma, we partition $\mathbf{y}$ into $b:=\lceil m/5\rceil$ consecutive blocks of length at most $5$ and consider all ordered triples of distinct positions within each block. For example, if a block is $(y_1,\dots,y_5)$, then we consider all triples $(y_i,y_j,y_k)$ with distinct $i,j,k\in[5]$. Thus, the total number of triples considered is at most $3!\binom{5}{3}b=60b$.  

\begin{lemma}\label{lemma:many genuine triples}
    Let $\mathbf{y}\in\mathbb{F}_{q^3}^m$ be obtained from $\mathbf{c}_f$ by an arbitrary coordinate permutation followed by at most $n-3$ insdel errors. Partition $\mathbf{y}$ into $b=\lceil m/5\rceil$ consecutive blocks of length at most $5$. If all ordered triples of distinct positions inside each block are considered, then at least
    $$
    h:=\left\lceil\frac{m+3}{2}\right\rceil-2\left\lceil\frac{m}{5}\right\rceil
    $$
    of these triples are genuine triples.
\end{lemma}

\begin{proof}
    Let $g$ be the number of genuine surviving symbols in $\mathbf{y}$. Suppose that $\mathbf{y}$ is obtained from $\mathbf{c}_f$ by first permuting the codeword and then performing $d$ deletions and $i$ insertions, where $d+i\le n-3$. Then $g=n-d=m-i$, and hence
    $$
    g=\frac{n+m-(d+i)}{2}
    \ge
    \left\lceil\frac{m+3}{2}\right\rceil.
    $$

    For each $j\in[b]$, let $g_j$ be the number of genuine positions in the $j$-th block. Then $\sum_{j\in[b]}g_j=g$. Since all ordered triples inside each block are considered, every three genuine positions in the same block contribute at least one genuine triple, namely the ordering that agrees with their order in $\mathbf{c}_f$. Therefore, the number of genuine triples considered is at least
    $$
    \sum_{j\in[b]}\binom{g_j}{3}
    \ge
    \sum_{j\in[b]}(g_j-2)
    =
    g-2b
    \ge
    \left\lceil\frac{m+3}{2}\right\rceil
    -
    2\left\lceil\frac{m}{5}\right\rceil.$$ This completes the proof.
\end{proof}

We now introduce the first step of the decoding algorithm. Let $\mathsf{Seed}$ denote the inverse-ratio seed algorithm obtained from the deletion decoder in Lemma~\ref{lemma:2dim alg deletion}. Given an ordered triple of received symbols, $\mathsf{Seed}$ either outputs the corresponding affine polynomial or returns no output. As discussed in Section~\ref{section:Singhvi decoding}, each call to $\mathsf{Seed}$ takes $O(1)$ time.

We apply $\mathsf{Seed}$ to all ordered triples within the blocks described above. The affine polynomials returned by these calls are collected as candidates, and we keep only those that receive at least $h$ votes. We also use the following convention for the seed subroutine: if the input triple is $(a,a,a)$, then $\mathsf{Seed}$ outputs the constant polynomial $f(x)=a$. Otherwise, when the three symbols are not all equal, $\mathsf{Seed}$ applies the inverse-ratio procedure described above.

\begin{algorithm}[h]
\caption{Local candidates majority voting}
\label{alg:Local candidates majority voting}
\KwIn{Integer $m$ and $b=\left\lceil\frac{m}{5}\right\rceil$; received word $\mathbf{y}=(y_1,\dots,y_m)$; majority threshold $h$; inverse-ratio seed algorithm $\mathsf{Seed}$.}
\KwOut{A candidate list $\mathcal{F}$ of affine polynomials with at least $h$ votes.}

$\mathcal{F}\leftarrow \varnothing$. \\
Initialize an empty hash table $\mathsf{Vote}$. \\
Partition $\mathbf{y}$ into $b$ consecutive blocks of length at most $5$. \\
\For{each block $W$}{
    \For{each $3$-subset of distinct positions $\lambda_1,\lambda_2,\lambda_3$ in $W$}{
        \For{each permutation $\pi$ of $[3]$}{
            \If{$\mathsf{Seed}(y_{\lambda_{\pi(1)}},y_{\lambda_{\pi(2)}},y_{\lambda_{\pi(3)}})$ returns an affine polynomial $\hat{f}(x)$}{
                $\mathsf{Vote}[\hat{f}(x)]\gets \mathsf{Vote}[\hat{f}(x)]+1$. \\
            }
        }
    }
}
\For{each affine polynomial $\hat{f}(x)$ stored in $\mathsf{Vote}$}{
    \If{$\mathsf{Vote}[\hat{f}(x)]\ge h$}{
        $\mathcal{F}\leftarrow \mathcal{F}\cup \{\hat{f}(x)\}$. \\
    }
}
Output $\mathcal{F}$.
\end{algorithm}

\begin{lemma}[Behavior of Algorithm~\ref{alg:Local candidates majority voting}]\label{lemma:behavior alg majority voting}
    For any transmitted codeword $\mathbf{c}_f$ corresponding to the polynomial $f(x)$, and for any possible received word $\mathbf{y}$, Algorithm~\ref{alg:Local candidates majority voting} outputs a list $\mathcal{F}$ of size $O(1)$ such that $f(x)\in\mathcal{F}$. Moreover, Algorithm~\ref{alg:Local candidates majority voting} runs in $O(n)$ time.
\end{lemma}

\begin{proof}
    By Lemma~\ref{lemma:many genuine triples}, at least $h$ ordered triples considered by the algorithm are genuine triples. For each such genuine triple, the inverse-ratio seed algorithm $\mathsf{Seed}$ outputs the transmitted polynomial $f(x)$. Hence $f(x)$ receives at least $h$ votes, and therefore $f(x)\in\mathcal{F}$.

 Next, we bound the size of $\mathcal{F}$. The algorithm considers at most $60b$ ordered triples. If $m=3$, then $|\mathcal{F}|\le 6=O(1)$. Thus assume $m\ge 4$, so that $h>0$. Since only polynomials with at least $h$ votes are included in $\mathcal{F}$, we have
    $$
    |\mathcal{F}|
    \le
    \frac{60b}{h}
    =
    \frac{60\left\lceil m/5\right\rceil}
    {\left\lceil (m+3)/2\right\rceil-2\left\lceil m/5\right\rceil}
    =
    O(1).
    $$
    Finally, since at most $60b=O(n)$ ordered triples are considered and each call to $\mathsf{Seed}$ takes $O(1)$ time, the total running time is $O(n)$.
\end{proof}

Next, we need to identify the transmitted polynomial $f(x)$ from the candidate list $\mathcal{F}$. A direct approach is to test, for each $\hat{f}(x)\in\mathcal{F}$, whether the received word $\mathbf{y}$ can be obtained from a permutation of $\mathbf{c}_{\hat f}$ by at most $n-3$ insdel errors. Computing this by a standard edit-distance algorithm would take $O(n^2)$ time~\cite{Bringmann2024faster}.

When $\hat f(x)$ is non-constant, the entries of the codeword $\mathbf{c}_{\hat f}$ are pairwise distinct. Thus it suffices to test which distinct symbols of $\mathbf{y}$ appear in $\mathbf{c}_{\hat f}$. The following algorithm does this in average $O(n)$ time by building a hash table for the evaluation points and using average $O(1)$-time membership queries.

\begin{algorithm}[h]
\caption{Find the correct codeword in $\mathcal{F}$.}
\label{alg:Find correct in list}
\KwIn{Integer $m$ and $\tau:=\left\lceil\frac{m+3}{2}\right\rceil$; received word $\mathbf{y}=(y_1,\dots,y_m)$; evaluation points $\boldsymbol{\alpha}=(\alpha_1,\dots,\alpha_n)$; candidate list $\mathcal{F}$.}
\KwOut{The transmitted codeword $\mathbf{c}_f$ and the corresponding polynomial $f(x)$.}

Treat $\mathbf{y}$ as a multiset and let $\mathbf{y}'=\{y_1',\dots,y_w'\}$ be the set of distinct symbols appearing in $\mathbf{y}$. \\ 
Construct a hash table storing the evaluation points $\alpha_1,\dots,\alpha_n$. \\
\For{each $\hat{f}(x)=m_0+m_1x\in \mathcal{F}$}{
    \uIf{$\hat{f}(x)$ is constant}{
        \If{$m_0$ appears at least $\tau$ times in $\mathbf{y}$}{
            $f(x)\leftarrow\hat{f}(x)$, $\mathbf{c}_{f}\leftarrow f(\boldsymbol{\alpha})$. \\
            Output $\mathbf{c}_{f}$, $f(x)$ and halt.\\
        }
    }
    \Else{
        $p\leftarrow 0$. Here $p$ indicates the number of common symbols.\\ 
        \For{$j\leftarrow 1$ \KwTo $w$}{
            $\hat{\alpha}\leftarrow (y_j'-m_0)/m_1$. \\
            \If{$\hat{\alpha}$ appears in the hash table}{
                $p\leftarrow p+1$. \\
            }
        }
        \If{$p\ge \tau$}{
            $f(x)\leftarrow\hat{f}(x)$, $\mathbf{c}_{f}\leftarrow f(\boldsymbol{\alpha})$. \\
            Output $\mathbf{c}_{f}$, $f(x)$ and halt.\\
        }
    }  
}
\end{algorithm}

\begin{lemma}[Behavior of Algorithm~\ref{alg:Find correct in list}]\label{lemma:behavior alg find in list}
    For any received word $\mathbf{y}$, with input $\mathcal{F}$ given by Algorithm~\ref{alg:Local candidates majority voting}, Algorithm~\ref{alg:Find correct in list} outputs a unique codeword $\mathbf{c}_f$ and the corresponding polynomial $f(x)$ such that $\mathbf{y}$ can be obtained from a permutation of $\mathbf{c}_f$ by at most $n-3$ insdel errors. Moreover, Algorithm~\ref{alg:Find correct in list} runs in average $O(n)$ time.
\end{lemma}

\begin{proof}
    For a candidate codeword $\mathbf{c}_{\hat f}$, allowing an arbitrary coordinate permutation means that we only need to consider the multiset overlap between $\mathbf{c}_{\hat f}$ and $\mathbf{y}$. By the identity $\mathrm{ED}(s,s')=|s|+|s'|-2\mathrm{LCS}(s,s')$, the word $\mathbf{y}$ can be obtained from some permutation of $\mathbf{c}_{\hat f}$ using at most $n-3$ insdel errors if and only if the two multisets share at least $\tau=\lceil(m+3)/2\rceil$ symbols.

    Algorithm~\ref{alg:Find correct in list} computes exactly this multiset overlap. In the constant case, $\mathbf{c}_{\hat f}$ consists only of the symbol $m_0$, so the algorithm checks whether $m_0$ appears at least $\tau$ times in $\mathbf{y}$. In the nonconstant case, the entries of $\mathbf{c}_{\hat f}$ are pairwise distinct. Hence it suffices to count how many distinct symbols of $\mathbf{y}$ appear in $\mathbf{c}_{\hat f}$. For a symbol $y_j'$, we have $y_j'\in\mathbf{c}_{\hat f}$ if and only if $(y_j'-m_0)/m_1$ is one of the evaluation points $\alpha_1,\dots,\alpha_n$, which is exactly what the hash-table test checks.

    Therefore, the algorithm outputs precisely those candidate codewords whose multiset overlap with $\mathbf{y}$ is at least $\tau$. Since Algorithm~\ref{alg:Local candidates majority voting} ensures that the transmitted polynomial $f(x)$ belongs to $\mathcal{F}$, at least one candidate passes this test. Uniqueness follows from Construction~\ref{constr:2dim}, because the code is robust against the $(n-3)$-permutation-insdel adversary.

    Finally, we analyze the running time. Since $|\mathcal{F}|=O(1)$, it suffices to show that each candidate is tested in average $O(n)$ time. The constant case is immediate. In the non-constant case, the algorithm scans the set $\mathbf{y}'$ of distinct received symbols, whose size is at most $m=O(n)$, and each hash-table membership query takes average $O(1)$ time. Thus each candidate is tested in average $O(n)$ time, which completes the proof.
\end{proof}

By Lemma~\ref{lemma:behavior alg majority voting} and Lemma~\ref{lemma:behavior alg find in list}, Algorithms~\ref{alg:Local candidates majority voting} and~\ref{alg:Find correct in list} together complete the proof of Theorem~\ref{thm:2dim decode informal}, which we restate for convenience.

\begin{thm}\label{thm:2dim decode perm insdel}
    Let $\mathcal{C}=\textup{RS}_{n,2}(\boldsymbol{\alpha})$ be the code over
    $\mathbb{F}_{q^3}$ specified in Construction~\ref{constr:2dim}. Then there exists a decoder that recovers the transmitted codeword from any received word obtained by applying an arbitrary coordinate permutation followed by at most $n-3$ insdel errors in average $O(n)$ time.
\end{thm}

\begin{rmk}
    The argument can be viewed as using a deletion decoder as a seed subroutine in the high-deletion regime. Its efficiency relies on the fact that the number of deletions is close to the block length $n$, so a few genuine symbols are enough to identify a candidate codeword. More generally, the running time depends on the cost of the underlying deletion-decoding subroutine.
\end{rmk}

\section{Concluding Remarks}

By allowing an additive $\epsilon n$ gap from the half-Singleton bound, we reduce the alphabet size required for Reed--Solomon codes against permutation-insdel errors from exponential~\cite{Con2025anonymous} to polynomial in $n$. Our construction is randomized and relies on union bounds over the $\mathbf{V}$- and $\mathbf{A}$-matrix regimes. We also prove a general histogram-packing alphabet-size lower bound for $q$-ary codes in the permutation-insdel setting, which in particular rules out linear-size alphabets for constant-rate Reed--Solomon codes. Finally, for the previously known alphabet-optimal two-dimensional Reed--Solomon construction, we give an average linear-time decoder in the permutation-insdel model. We conclude with several natural directions for future work.

\begin{itemize}
    \item[1.] \emph{Improving the alphabet size.}
    In the ordinary insdel setting, random Reed--Solomon codes can approach the half-Singleton bound over linear-size alphabets~\cite{Con2024random}. In the permutation-insdel setting, our current approach incurs substantial combinatorial losses: the $\mathbf{V}$-matrix regime contains a factorial term $\ell!$, while the $\mathbf{A}$-matrix regime contains a term of order $t^{\ell+1}$. These terms arise from counting possible index orders and set patterns, and appear to be inherent to the present $\mathbf{V}$-matrix/$\mathbf{A}$-matrix framework. A substantial further reduction of the alphabet size would likely require a different framework, rather than only sharper estimates within the same counting argument.

    \item[2.] \emph{Alphabet-size lower bounds.}
    Although we establish an alphabet-size lower bound, we believe that it is far from tight, since the proof uses only a histogram packing argument and does not exploit the algebraic structure of Reed--Solomon codes. It remains interesting to derive sharper lower bounds by using some algebraic properties.

    \item[3.] \emph{Explicit constructions.}
    The main existence result in this paper is probabilistic. A natural next step is to construct explicit evaluation sets over alphabets of polynomial size that achieve the same near-half-Singleton guarantee.

    \item[4.] \emph{Efficient decoding algorithms.}
    Efficient decoding for general Reed--Solomon codes under insdel errors remains challenging, especially for parameters near the half-Singleton bound. Developing efficient decoders for the permutation-insdel model is an interesting direction.
\end{itemize}

\bibliographystyle{IEEEtran}
\bibliography{RSperm}

\begin{thebibliography}{10}
\providecommand{\url}[1]{#1}
\csname url@samestyle\endcsname
\providecommand{\newblock}{\relax}
\providecommand{\bibinfo}[2]{#2}
\providecommand{\BIBentrySTDinterwordspacing}{\spaceskip=0pt\relax}
\providecommand{\BIBentryALTinterwordstretchfactor}{4}
\providecommand{\BIBentryALTinterwordspacing}{\spaceskip=\fontdimen2\font plus
\BIBentryALTinterwordstretchfactor\fontdimen3\font minus \fontdimen4\font\relax}
\providecommand{\BIBforeignlanguage}[2]{{%
\expandafter\ifx\csname l@#1\endcsname\relax
\typeout{** WARNING: IEEEtran.bst: No hyphenation pattern has been}%
\typeout{** loaded for the language `#1'. Using the pattern for}%
\typeout{** the default language instead.}%
\else
\language=\csname l@#1\endcsname
\fi
#2}}
\providecommand{\BIBdecl}{\relax}
\BIBdecl

\bibitem{Reed1960polynomial}
I.~S. Reed and G.~Solomon, ``Polynomial codes over certain finite fields,'' \emph{J. Soc. Indust. Appl. Math.}, vol.~8, no.~2, pp. 300--304, 1960.

\bibitem{xia2015tale}
M.~Xia, M.~Saxena, M.~Blaum, and D.~A. Pease, ``A tale of two erasure codes in {HDFS},'' in \emph{Proc. 13th USENIX Conf. File Storage Technol. (FAST)}, Santa Clara, CA, USA, 2015, pp. 213--226.

\bibitem{Rashmi2013solution}
K.~V. Rashmi, N.~B. Shah, D.~Gu, H.~Kuang, D.~Borthakur, and K.~Ramchandran, ``A solution to the network challenges of data recovery in erasure-coded distributed storage systems: A study on the facebook warehouse cluster,'' in \emph{Proc. 5th USENIX Workshop Hot Topics Storage File Syst. (HotStorage)}, San Jose, CA, USA, 2013, p.~8.

\bibitem{Levenshtein1965binary}
V.~I. Levenshtein, ``Binary codes capable of correcting deletions, insertions, and reversals,'' \emph{Soviet Physics Dokl.}, vol.~10, no.~8, pp. 707--710, 1966.

\bibitem{Guruswami2021twodeletion}
V.~Guruswami and J.~H{\aa}stad, ``Explicit two-deletion codes with redundancy matching the existential bound,'' \emph{IEEE Trans. Inform. Theory}, vol.~67, no.~10, pp. 6384--6394, 2021.

\bibitem{Sun2024twoedit}
Y.~Sun and G.~Ge, ``Binary codes for correcting two edits,'' \emph{IEEE Trans. Inform. Theory}, vol.~70, no.~10, pp. 6877--6898, 2024.

\bibitem{Brakensiek2018multipledeletion}
J.~Brakensiek, V.~Guruswami, and S.~Zbarsky, ``Efficient low-redundancy codes for correcting multiple deletions,'' \emph{IEEE Trans. Inform. Theory}, vol.~64, no.~5, pp. 3403--3410, 2018.

\bibitem{Haeupler2021synchronization}
B.~Haeupler and A.~Shahrasbi, ``Synchronization strings: codes for insertions and deletions approaching the singleton bound,'' \emph{J. ACM}, vol.~68, no.~5, pp. Art. 36, 39, 2021.

\bibitem{Cheng2022deterministic}
K.~Cheng, Z.~Jin, X.~Li, and K.~Wu, ``Deterministic document exchange protocols and almost optimal binary codes for edit errors,'' \emph{J. ACM}, vol.~69, no.~6, pp. Art. 44, 39, 2022.

\bibitem{Abdel2010correcting}
K.~A.~S. Abdel-Ghaffar, H.~C. Ferreira, and L.~Cheng, ``Correcting deletions using linear and cyclic codes,'' \emph{IEEE Trans. Inform. Theory}, vol.~56, no.~10, pp. 5223--5234, 2010.

\bibitem{Cheng2023efficient}
K.~Cheng, V.~Guruswami, B.~Haeupler, and X.~Li, ``Efficient linear and affine codes for correcting insertions/deletions,'' \emph{SIAM J. Discrete Math.}, vol.~37, no.~2, pp. 748--778, 2023.

\bibitem{con2022explicit}
R.~Con, A.~Shpilka, and I.~Tamo, ``Explicit and efficient constructions of linear codes against adversarial insertions and deletions,'' \emph{IEEE Trans. Inform. Theory}, vol.~68, no.~10, pp. 6516--6526, 2022.

\bibitem{xie2024new}
C.~Xie, H.~Chen, L.~Qu, and L.~Liu, ``New dimension-independent upper bounds on linear insdel codes,'' \emph{Adv. Math. Commun.}, vol.~18, no.~6, pp. 1575--1589, 2024.

\bibitem{Con2023reed}
R.~Con, A.~Shpilka, and I.~Tamo, ``Reed {S}olomon codes against adversarial insertions and deletions,'' \emph{IEEE Trans. Inform. Theory}, vol.~69, no.~5, pp. 2991--3000, 2023.

\bibitem{Con2024random}
R.~Con, Z.~Guo, R.~Li, and Z.~Zhang, ``Random {Reed--Solomon} codes achieve the {Half-Singleton} bound for insertions and deletions over linear-sized alphabets,'' in \emph{Proc. Int. Colloq. Automata Lang. Program. (ICALP)}, Aarhus, Denmark, 2025, pp. 60:1--60:21.

\bibitem{Tonien2007construction}
D.~Tonien and R.~Safavi-Naini, ``Construction of deletion correcting codes using generalized {R}eed-{S}olomon codes and their subcodes,'' \emph{Des. Codes Cryptogr.}, vol.~42, no.~2, pp. 227--237, 2007.

\bibitem{DoDuc2021explicit}
T.~Do~Duc, S.~Liu, I.~Tjuawinata, and C.~Xing, ``Explicit constructions of two-dimensional {R}eed-{S}olomon codes in high insertion and deletion noise regime,'' \emph{IEEE Trans. Inform. Theory}, vol.~67, no.~5, pp. 2808--2820, 2021.

\bibitem{Liu2021twodimension}
S.~Liu and I.~Tjuawinata, ``On 2-dimensional insertion-deletion {R}eed-{S}olomon codes with optimal asymptotic error-correcting capability,'' \emph{Finite Fields Appl.}, vol.~73, p. Art. No. 101841, 2021.

\bibitem{Liu2024optimal}
J.~Liu, ``Optimal {RS} codes and {GRS} codes against adversarial insertions and deletions and optimal constructions,'' \emph{IEEE Trans. Inform. Theory}, vol.~70, no.~9, pp. 6269--6279, 2024.

\bibitem{Beelen2026reed}
P.~Beelen, R.~Con, A.~Gruica, M.~Montanucci, and E.~Yaakobi, ``Reed-{S}olomon codes against insertions and deletions: full-length and rate-1/2 codes,'' \emph{IEEE Trans. Inform. Theory}, vol.~72, no.~2, pp. 1149--1160, 2026.

\bibitem{Con2024twodimension}
R.~Con, A.~Shpilka, and I.~Tamo, ``Optimal two-dimensional {R}eed-{S}olomon codes correcting insertions and deletions,'' \emph{IEEE Trans. Inform. Theory}, vol.~70, no.~7, pp. 5012--5016, 2024.

\bibitem{Wang2004deletion}
Y.~Wang, L.~McAven, and R.~Safavi-Naini, ``Deletion correcting using generalized {Reed--Solomon} codes,'' in \emph{Coding, Cryptography and Combinatorics}.\hskip 1em plus 0.5em minus 0.4em\relax Cham, Switzerland: Springer, 2004, pp. 345--358.

\bibitem{Con2025anonymous}
R.~Con, ``Anonymous {S}hamir's {S}ecret-{S}haring via {R}eed--{S}olomon {C}odes {A}gainst {P}ermutations, {I}nsertions, and {D}eletions,'' \emph{IEEE Trans. Inform. Theory}, vol.~71, no.~12, pp. 9534--9547, 2025.

\bibitem{Shamir1979how}
A.~Shamir, ``How to share a secret,'' \emph{Comm. ACM}, vol.~22, no.~11, pp. 612--613, 1979.

\bibitem{Harry2024abuse}
H.~Eldridge, G.~Beck, M.~Green, N.~Heninger, and A.~Jain, ``{Abuse-resistant} location tracking: Balancing privacy and safety in the offline finding ecosystem,'' in \emph{Proc. 33rd USENIX Security Symp. (USENIX Security)}, Philadelphia, PA, USA, 2024, pp. 5431--5448.

\bibitem{Bishop2025fully}
A.~Bishop, M.~Green, Y.~Ishai, A.~Jain, and P.~Lou, ``Fully anonymous secret sharing,'' in \emph{Proc. Annu. Int. Cryptol. Conf. (CRYPTO)}, Santa Barbara, CA, USA, 2025, pp. 356--389.

\bibitem{Singhvi2026twodimension}
S.~Singhvi, ``Optimally decoding 2-{D} {R}eed-{S}olomon codes against deletion errors,'' \emph{IEEE Trans. Inform. Theory}, vol.~72, no.~3, pp. 1646--1653, 2026.

\bibitem{Banerjee2025decoding}
A.~Banerjee, R.~Con, A.~Wachter-Zeh, and E.~Yaakobi, ``Decoding insertions/deletions via list recovery,'' in \emph{Proc. IEEE Int. Symp. Inf. Theory (ISIT)}, Ann Arbor, MI, USA, 2025, pp. 1--6.

\bibitem{Demillo1977probabilistic}
R.~A. DeMillo and R.~J. Lipton, ``{A} probabilistic remark on algebraic program testing,'' \emph{Inf. Process. Lett.}, vol.~7, no.~4, pp. 193--195, 1978.

\bibitem{Zippel1979probabilistic}
R.~Zippel, ``Probabilistic algorithms for sparse polynomials,'' in \emph{Proc. Int. Symp. Symbolic Algebraic Comput. (EUROSAM)}, Marseille, France, 1979, pp. 216--226.

\bibitem{Schwartz1980fast}
J.~T. Schwartz, ``Fast probabilistic algorithms for verification of polynomial identities,'' \emph{J. Assoc. Comput. Mach.}, vol.~27, no.~4, pp. 701--717, 1980.

\bibitem{Bringmann2024faster}
K.~Bringmann, A.~Cassis, N.~Fischer, and T.~Kociumaka, ``Faster sublinear-time edit distance,'' in \emph{Proc. Annu. ACM-SIAM Symp. Discrete Algorithms (SODA)}, Alexandria, VA, USA, 2024, pp. 3274--3301.

\end{thebibliography}
\end{document}